\newif\ifacm
\newif\ifnotacm
\acmfalse

\ifacm
\notacmfalse
\else
\notacmtrue
\fi

\ifacm
\documentclass[prodmode,acmjacm]{acmsmall}
\else
\documentclass[11pt]{article}
\fi
\usepackage{amsmath, amssymb}
\ifnotacm
\usepackage{amsthm}
\usepackage{fullpage}
\usepackage[margin=1in]{geometry}
\fi
\usepackage{comment}
\usepackage{url}
\usepackage{xcolor}
\usepackage{tikz}
\usepackage{enumerate}
\usepackage{mdwlist}
\usepackage[pdftex,colorlinks=true,linkcolor=blue,citecolor=blue,urlcolor=black]{hyperref}

\newcommand{\bra}[1]{\langle #1|}
\newcommand{\ket}[1]{|#1\rangle}
\newcommand{\proj}[1]{|#1\rangle\langle #1|}
\newcommand{\braket}[2]{\langle #1|#2\rangle}
\newcommand{\hsip}[2]{\langle #1,#2 \rangle}

\DeclareMathOperator{\conv}{conv}
\DeclareMathOperator{\poly}{poly}
\DeclareMathOperator{\polylog}{polylog}
\DeclareMathOperator{\tr}{tr}

\DeclareMathOperator{\supp}{supp}
\DeclareMathOperator{\Sym}{Sym}
\DeclareMathOperator{\Sep}{Sep}
\DeclareMathOperator{\SepSym}{SepSym}
\DeclareMathOperator{\ProdSym}{ProdSym}
\DeclareMathOperator{\SEP}{SEP}
\DeclareMathOperator{\PPT}{PPT}
\DeclareMathOperator{\WMEM}{WMEM}

\DeclareMathOperator{\BPP}{\mathsf{BPP}}
\DeclareMathOperator{\BQP}{\mathsf{BQP}}
\DeclareMathOperator{\DTIME}{\mathsf{DTIME}}
\DeclareMathOperator{\NTIME}{\mathsf{NTIME}}
\DeclareMathOperator{\MA}{\mathsf{MA}}
\DeclareMathOperator{\NP}{\mathsf{NP}}
\DeclareMathOperator{\NEXP}{\mathsf{NEXP}}
\DeclareMathOperator{\Ptime}{\mathsf{P}}
\DeclareMathOperator{\QMA}{\mathsf{QMA}}

\DeclareMathOperator{\BellQMA}{\mathsf{BellQMA}}

\def\sat{\text{3-SAT}}
\def\be#1\ee{\begin{equation}#1\end{equation}}
\def\bea#1\eea{\begin{eqnarray}#1\end{eqnarray}}
\def\beas#1\eeas{\begin{eqnarray*}#1\end{eqnarray*}}
\def\ba#1\ea{\begin{align}#1\end{align}}
\def\bas#1\eas{\begin{align*}#1\end{align*}}
\def\nn{\nonumber}
\def\eq#1{Eq.~(\ref{eq:#1})}
\def\bit{\begin{itemize}}
\def\eit{\end{itemize}}
\def\L{\left} 
\def\R{\right}
\def\ra{\rightarrow}
\def\ot{\otimes}

\newtheorem{thm}{Theorem}
\ifnotacm
\newtheorem*{thm*}{Theorem}
\fi

\newtheorem{cor}[thm]{Corollary}
\newtheorem{lem}[thm]{Lemma}
\newtheorem{prop}[thm]{Proposition}
\newtheorem{dfn}{Definition}
\newtheorem{proto}{Protocol}

\makeatletter
\ifacm
\newtheorem{rep@theorem}{\rep@title}
\else
\newtheorem*{rep@theorem}{\rep@title}
\fi
\newcommand{\newreptheorem}[2]{%
\newenvironment{rep#1}[1]{%
 \def\rep@title{#2 \ref{##1} (restatement)}%
 \begin{rep@theorem}[]}%
 {\end{rep@theorem}}}
\makeatother

\newreptheorem{thm}{Theorem}
\newreptheorem{lem}{Lemma}

\ifnotacm
\theoremstyle{remark}
\newtheorem*{remark}{Remark}
\fi

\def\eps{\epsilon}

\def\cB{\mathcal{B}}
\def\cD{\mathcal{D}}
\def\cF{\mathcal{F}}

\def\cN{\mathcal{N}}

\def\cS{\mathcal{S}}

\def\bbC{\mathbb{C}}
\def\bbE{\mathbb{E}}
\def\bbM{\mathbb{M}}
\def\bbN{\mathbb{N}}
\def\bbR{\mathbb{R}}

\def\benum{\begin{enumerate}}
\def\eenum{\end{enumerate}}
\def\bit{\begin{itemize}}
\def\eit{\end{itemize}}

\newcommand{\secref}[1]{Section~\ref{sec:#1}}
\newcommand{\appref}[1]{Appendix~\ref{sec:#1}}
\newcommand{\lemref}[1]{Lemma~\ref{lem:#1}}
\newcommand{\thmref}[1]{Theorem~\ref{thm:#1}}

\newcommand{\protoref}[1]{Protocol~\ref{proto:#1}}

\newcommand{\corref}[1]{Corollary~\ref{cor:#1}}

\newcommand{\boxproto}[2]{
\begin{figure}[h]
\begin{center}
\noindent \framebox{
\begin{minipage}{0.8\textwidth}
\begin{proto}[{\bf #1}]
\ \\ \\
#2
\end{proto}
\end{minipage}
}
\end{center}
\end{figure}
}

\def\Ptest{P_{\text{test}}}
\DeclareMathOperator{\Pprod}{OPP}

\begin{document}

\hfuzz=6pt

\title{Testing product states, quantum Merlin-Arthur games and tensor optimisation}

\ifacm
\author{ARAM W. HARROW \affil{University of Washington} and ASHLEY  MONTANARO \affil{University of Cambridge}}
\markboth{A. W. Harrow and A. Montanaro}{Testing product states}
\else
\author{Aram W.\ Harrow\footnote{Department of Computer Science \&  Engineering, University of Washington; {\tt aram@cs.washington.edu}.}\; and Ashley Montanaro\footnote{Department of Applied Mathematics and Theoretical Physics, University of Cambridge; {\tt am994@cam.ac.uk}.}}
\maketitle
\fi

\begin{abstract}
  We give a test that can distinguish efficiently between product
  states of $n$ quantum systems and states which are far from
  product. If applied to a state $\ket{\psi}$ whose maximum overlap
  with a product state is $1-\epsilon$, the test passes with
  probability $1-\Theta(\epsilon)$, regardless of $n$ or the local
  dimensions of the individual systems.  The test uses two copies of
  $\ket{\psi}$. We prove correctness of this test as a special case of
  a more general result regarding stability of maximum output purity
  of the depolarising channel.

  A key application of the test is to quantum Merlin-Arthur games with
  multiple Merlins, where we obtain several structural results that had been previously conjectured,
  including the fact that efficient soundness amplification is possible and that
  two Merlins can simulate many Merlins: $\QMA(k)=\QMA(2)$ for $k\ge 2$.
  Building on a previous
  result of Aaronson et al., this implies that there is an efficient
  quantum algorithm to verify $\sat$ with constant soundness, given
  two unentangled proofs of $\widetilde{O}(\sqrt{n})$ qubits. 
We also show how $\QMA(2)$ with log-sized proofs is equivalent to a
large number of problems, some related to quantum information (such as testing separability of mixed
states) as well as problems without any apparent connection to quantum
mechanics (such as computing injective tensor norms of 3-index
tensors).  As a consequence, we obtain many hardness-of-approximation
results, as well as potential algorithmic applications of methods for
approximating $\QMA(2)$ acceptance probabilities.
%Our result can also be interpreted as a reduction from the problem of
%computing the injective tensor norm of a multiple-index tensor to
%that of a 3-index tensor.   

  Finally, our test can also be used to construct an efficient test
  for determining whether a unitary operator is a tensor product,
  which is a generalisation of classical linearity testing.
\end{abstract}

\ifacm
\maketitle
\fi

%-------------------------------------------------------------------------------

\section{Introduction}
\label{sec:introduction}

Entanglement of quantum states presents both an opportunity and a
difficulty for quantum computing. To describe a pure state of $n$
qudits ($d$-dimensional quantum systems) requires a comparable number
of parameters to a classical probability distribution on $d^n$ items.
Effective methods are known for testing properties of probability
distributions. However, for quantum states many of these tools no
longer work. For example, due to interference, the probability of a
test passing cannot be simply written as an average over components of
the state.  Moreover, measuring one part of a state and conditioning on
the measurement outcome may induce
entanglement between other parts of the state that were not previously
entangled with each other.

These counter-intuitive properties of entanglement account for many of
the outstanding puzzles in quantum information.  In quantum
channel coding, the famous additivity violations of
\cite{DSS97,Hastings} reflect how entangled inputs can sometimes have
advantages against even uncorrelated noise.  For quantum interactive
proofs, the primary difficulty is in bounding the ability of provers
to cheat using entangled strategies~\cite{IKM08}.  Even for $\QMA(k)$
(the variant of $\QMA$ with $k$ unentangled Merlins
\cite{kobayashi03,Unentanglement}), most important open questions could be
resolved by finding a way to control entanglement within each proof.
Here, the recently discovered failure of parallel repetition for
entangled provers~\cite{KR09} is a sort of complexity-theoretic
analogue of additivity violations.

The situation is different when we consider quantum states that are
{\em product} across the $n$ systems.  In this case, while individual
systems of course behave quantumly, the lack of correlation between
the systems means that classical tools such as Chernoff bounds can be
used.  For example, in channel coding with product-state inputs, not
only does the single-letter Holevo formula give the capacity, so that
there is no additivity problem, but so-called strong converse theorems are
known, which prove that attempting to communicate at a rate above the
capacity results in an exponentially decreasing probability of
successfully transmitting a message \cite{ON99,winter99}.  Naturally,
many of the difficulties in dealing with entangled proofs and quantum
parallel repetition would also go away if quantum states were
constrained to be in product form.

%-------------------------------------------------------------------------------

\subsection{Our results}

In this paper, we present a quantum test to determine whether an
$n$-partite state $\ket{\psi}$ is a product state or far from any
product state.  We make no assumptions about the local dimensions of
$\ket{\psi}$; in fact, the local dimension can even be different for
different systems.  The test passes with certainty if $\ket{\psi}$ is
product, and fails with probability $\Theta(\epsilon)$ if the overlap
between $\ket{\psi}$ and the closest product state is $1-\epsilon$.
An essential feature of our test (or any possible such test, as we
will argue in Section \ref{sec:optimal}) is that it requires two
copies of $\ket{\psi}$.

The parameters of our test resemble classical property-testing
algorithms~\cite{fischer01}. In general, these algorithms make a small number of queries to
some object and accept with high probability if the object has some property
$P$ ({\em completeness}), and with low probability if the object is ``far''
from having property $P$ ({\em soundness}).  Crucially, the number of queries used and
the success probability should not depend on the size of the object.
The main result of this paper is a test for a property of a quantum
state, in contrast to previous work on quantum generalisations of
property testing, which has considered quantum algorithms for testing
properties of classical (e.g.\ \cite{property-jnl,atici07})
and quantum~\cite{qboolean} oracles (a.k.a.\ unitary operators,
although see \secref{unitaries} for an application to this setting).
In this sense, our work is closer to a body
of research on determining properties of quantum states directly,
without performing full tomography (e.g.\  the ``pretty good tomography'' of
Aaronson~\cite{Aaronson07}).
The direct detection of quantities relating to entanglement has
received particular attention; see~\cite{guhne09} for an extensive
review.  However, previous work has generally focused on Bell
inequalities and entanglement witnesses, which are typically designed
to distinguish a {\em particular} entangled state from any separable
state.  By contrast, our product test is generic and will detect
entanglement in any entangled state $\ket\psi$.

The product test is defined in \protoref{prodtest} below, and illustrated
schematically in Figure \ref{fig:prodtest}.
It uses as a subroutine the {\em swap test} for comparing quantum states
\cite{buhrman01}. This test, which can be implemented efficiently,
takes two (possibly mixed) states $\rho$, $\sigma$ of equal dimension as input. The test uses an ancilla qubit initialised in state $\ket{0}$ and applies a Hadamard gate to this qubit to produce the state $\proj{+}\otimes \rho \otimes \sigma$. The test proceeds by applying a controlled-SWAP operation to the latter two registers, controlled by the ancilla qubit, then applies a Hadamard gate on the ancilla qubit, followed by a computational basis measurement. If the outcome is 0, the output of the test is ``same''; otherwise, the output is ``different''. It is easy to show that this test outputs ``same'' with probability $\frac{1}{2} + \frac{1}{2}\tr \rho\,\sigma$.

\boxproto{Product test}{
\label{proto:prodtest}
The product test proceeds as follows.
\begin{enumerate}
\item Prepare two copies of $\ket{\psi} \in \bbC^{d_1} \ot \cdots \ot \bbC^{d_n}$; call these $\ket{\psi_1}$,
  $\ket{\psi_2}$. 
\item Perform the swap test on each of the $n$ pairs of corresponding subsystems of
  $\ket{\psi_1}$, $\ket{\psi_2}$. 
\item If all of the tests returned ``same'', accept. Otherwise, reject. 
\end{enumerate}
}

The product test has appeared before in the literature. It
was originally introduced in \cite{mintert05} as one of a family of
tests for generalisations of the concurrence entanglement measure, and
has been implemented experimentally as a means of detecting bipartite
entanglement directly \cite{walborn06} (but see also \cite{vELK07} for
caveats).  Further, the test was
proposed in \cite{qboolean} as a means of determining whether a
unitary operator is product.
%Indeed, there is a sense (which we will make precise in \secref{optimal}) in which \protoref{prodtest} is the canonical test to determine whether $\ket\psi$ is product when we are given two copies of the state.
Our contribution here is to prove the
correctness of this test for all $n$, as formalised in the following theorem.

\begin{thm}
\label{thm:prodtest}
Given $\ket{\psi}\in B(\bbC^{d_1} \ot \cdots \ot \bbC^{d_n})$, let
\[ 1 - \eps = \max\{ |\braket{\psi}{\phi_1,\ldots,\phi_n}|^2 :
\ket{\phi_i}\in B(\bbC^{d_i}), 1 \le i \le n\}.\]
Let $\Ptest(\proj{\psi})$ be the probability that the product test passes when applied to $\ket{\psi}$. Then 
\[ 1 - 2\eps + \eps^2 \le \Ptest(\proj{\psi}) \le 1 - \eps + \eps^2 +
\eps^{3/2}.\]
For some values of $\eps$ this bound is trivial, but we have the
following weaker bound which applies everywhere:
\be \Ptest(\proj{\psi}) \le 1 - \frac{11}{512}\eps.  \label{eq:prod-test-linear-bound}
\ee
%
%Furthermore, if $\eps \ge 11/32 > 0.343$, $\Ptest(\proj{\psi}) \le 501/512 < 0.979$.
More concisely, $\Ptest(\proj\psi) = 1-\Theta(\eps)$.
\end{thm}

This result is essentially best possible, in a number of ways. First,
we show in Section \ref{sec:optimal} that the product test itself is
optimal: among all tests for product states that use two copies and
have perfect completeness, the product test has optimal soundness. We
also show that there cannot exist any non-trivial test that uses only
one copy of the test state. Second, our analysis of the test cannot be
improved too much, without introducing dependence on $n$ and the local
dimensions. When $\epsilon$ is low, we give examples of states
$\ket{\psi}$ which achieve the upper and lower bounds on
$\Ptest(\proj{\psi})$, up to leading order. We also give an example of
a bipartite state for which $\epsilon$ is close to 1, but
$\Ptest(\proj{\psi}) \approx 1/2$, implying that the constant in our
bound cannot be replaced with a function of $\epsilon$ that goes to 0
as $\epsilon$ approaches 1. (The bounds on this constant obtained from
our proof could easily be improved somewhat, but we have not attempted
to do this.) See Appendix \ref{sec:prodtest} for all these
examples. Finally, it is unlikely that a similar test could be
developed for separability of {\em mixed} states, as the separability
problem for mixed states has been shown to be $\NP$-hard \cite{gurvits03,gharibian10} (and indeed we improve on this result, as discussed below).

The proof of Theorem \ref{thm:prodtest} is based on relating the probability of the test passing to
the action of the qudit depolarising channel. In fact, we prove a
considerably more general result regarding this channel. It is known
that the maximum output purity of this channel is
achieved for product state inputs \cite{AHW00}; our result, informally, says that
any state that is ``close'' to achieving maximum output purity must in
fact be ``close'' to a product state. This is a {\em stability} result
for this channel, which strengthens the previously known multiplicativity result.

\begin{figure}
\begin{center}
\begin{tikzpicture}[scale=1.25]

\filldraw[fill=gray!10,rounded corners] (-0.6,-0.4) rectangle (4.6,0.4);
\filldraw[fill=gray!10,rounded corners] (-0.6,0.6) rectangle (4.6,1.4);

\draw[rounded corners] (-0.4,-0.6) rectangle (0.4,1.6);
\draw[rounded corners] (0.6,-0.6) rectangle (1.4,1.6);
\draw[rounded corners] (1.6,-0.6) rectangle (2.4,1.6);
\draw[rounded corners] (3.6,-0.6) rectangle (4.4,1.6);

\foreach \x in {1,...,3} {
  \filldraw[fill=gray!75] (\x-1,0) node {\x} circle (0.25);
  \filldraw[fill=gray!50] (\x-1,1) node {\x} circle (0.25);
}
\node at (3,0) {\Huge ...};
\node at (3,1) {\Huge ...};
\filldraw[fill=gray!75] (4,0) node {$n$} circle (0.25);
\filldraw[fill=gray!50] (4,1) node {$n$} circle (0.25);
\node[anchor=east] at (-0.8,1) {$\ket{\psi_1}$};
\node[anchor=east] at (-0.8,0) {$\ket{\psi_2}$};
\end{tikzpicture}

\caption{Schematic of the product test applied to an $n$-partite state $\ket{\psi}$. The swap test (vertical boxes) is applied to the $n$ pairs of corresponding subsystems of two copies of $\ket{\psi}$ (horizontal boxes).}

\label{fig:prodtest}
\end{center}
\end{figure}
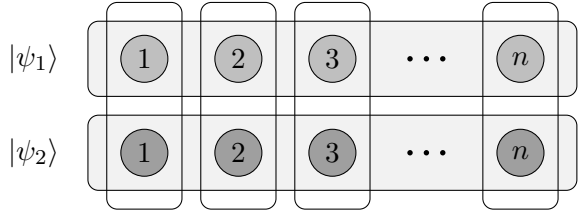

Somewhat more formally, let $\mathcal{D}_{\delta}$ be the
$d$-dimensional qudit depolarising channel with noise rate $1-\delta$, i.e.\
\be \label{eqn:depolarising} \mathcal{D}_{\delta} (\rho) = (1-\delta)(\tr \rho)\frac{I}{d} + \delta\,\rho \ee
for $\rho$ a arbitrary mixed state of one $d$-dimensional system, and
 define the {\em O}utput {\em P}urity of {\em P}roduct states to be
\be \Pprod(\delta) = \tr (\mathcal{D}_\delta^{\ot n}\, \proj{\phi})^2
\label{eq:OPP}\ee
where $\ket{\phi}$ is an arbitrary product state.   
Then our main result, stated more precisely as Theorem \ref{thm:stability}
in Section \ref{sec:depolar}, is that for small enough $\delta > 0$, if
$\tr (\mathcal{D}_\delta^{\ot n}\, \proj{\psi})^2 \ge
 (1-\epsilon)\Pprod(\delta)$,
then there is a product state $\ket{\phi_1,\dots,\phi_n}$ such that
$|\braket{\psi}{\phi_1,\dots,\phi_n}|^2 \geq 1 - O(\epsilon)$.

%-------------------------------------------------------------------------------

\subsection{Applications and interpretations of the product test}

We describe several applications of the product test.  The most
important of these is that this test can be used to relate $\QMA(k)$
to $\QMA(2)$, as we will discuss in \secref{QMA}.  The complexity
class $\QMA(k)$ is defined to be the class of languages that can be
decided with bounded error by a poly-time quantum verifier that
receives poly-size witnesses from $k$ unentangled
provers\footnote{We assume throughout this paper that
$k$ is at most polynomial in the input size.}~\cite{kobayashi03,Unentanglement}.  To put $\QMA(k)$ inside
$\QMA(2)$ with constant loss of soundness, we can have two provers
simulate $k$ provers by each submitting $k$ unentangled proofs, whose
lack of entanglement can be verified with our product test.  Indeed,
this gives an alternate way to understand our test as a method of
using bipartite separability to certify $k$-partite separability.

Surprisingly, using this result as a building block also
allows us to prove amplification for $\QMA(k)$ protocols.
It has been conjectured \cite{kobayashi03,Unentanglement}
that such protocols can be amplified to exponentially small soundness
error. We completely resolve this conjecture, showing that $\QMA(k)$
protocols can be simulated in $\QMA(2)$ with exponentially small soundness error,
and hence $\QMA(k) = \QMA(2)$ for $k\ge 2$. Indeed, we show that this result
still holds if the verifier's ``yes'' measurement operator in a $\QMA(2)$ protocol is required to be separable (see Appendix~\ref{sec:measure-defs} for the formal definition of this class of measurements).

As a further corollary, we can improve upon the
results of \cite{Unentanglement,BT09} to 
obtain a protocol in $\QMA(2)$ that verifies $\sat$ with constant soundness gap and
$O(\sqrt{n}\poly\log(n))$ qubits (where $n$ is the number of
clauses). This in turn allows us to prove hardness of approximation results for
19 problems from quantum information theory and elsewhere
which turn out to be closely related to $\QMA(2)$.  The complete list
of equivalent and related problems is given in \secref{equivalent}; while
most had previously been known, we believe that they had not been
previously collected in one place.

One example of such a problem is detecting separability, or in other words
the weak membership problem for $\Sep(d,d)$,
the set of separable quantum states on $d\times d$
dimensions.  It was shown in Ref.~\cite{gurvits03} that $\Sep$ cannot be
approximated to precision $\exp(-d)$ in time $\poly(d)$ unless $\Ptime=\NP$.
In Refs.~\cite{liu07,gharibian10}, this result was improved to
show that approximating $\Sep$ to precision $1/\poly(d)$ is similarly
$\NP$-hard.  We show that there is a universal constant $\delta>0$ such that,
if $K$ is a convex set that approximates
SEP to within trace distance $\delta$, then membership in $K$ cannot be
decided in polynomial time unless $\sat\in
\DTIME(\exp(\sqrt{n}\log^{O(1)}(n)))$.
Other such problems for which we can prove that no polynomial-time algorithm
exists, under the same assumption about the hardness of $\sat$, are estimating
the minimum output entropy of a quantum channel up to a constant, and
estimating the ground-state energy
of quantum systems under a mean-field approximation.

We also prove hardness results for some tensor optimisation problems which
are not apparently related directly to quantum information theory, examples of which
include approximating the injective tensor norm of 3-index tensors and estimating
the $\ell_2 \rightarrow \ell_4$ norm of a matrix. Our proof that amplification
of $\QMA(2)$ protocols is possible implies that one can derive stronger hardness results
for all of these tasks, if one is willing to make stronger assumptions about the hardness of $\sat$.

Our final application is that the product test can be used to
determine whether a unitary operator is a tensor product or far from a tensor product in the Hilbert-Schmidt norm, promised that one of these is the case.
This can be seen \cite{qboolean} as one possible generalisation of the
well-studied problem of testing whether a boolean function $\{0,1\}^n
\rightarrow \{0,1\}$ is linear~\cite{blum93}.  This application is
described in \secref{unitaries}.

These different applications of the product test reflect the many
different interpretations of $\Ptest(\proj\psi)$.  It is related in a precise sense to 
\bit
\item The purity of $\ket{\psi}$ after it is subjected to independent depolarising noise
(see Appendix \ref{sec:depolarising}).
\item The maximum overlap of $\ket{\psi}$ with any product state
(proved in Appendix \ref{sec:prodtest}). 
The logarithm of this maximum overlap is an important entanglement measure known as the geometric
measure of entanglement (see \cite{wei03} and references therein).
\item The overlap of $\ket{\psi}^{\ot 2}$ with the tensor product of
  the symmetric subspaces of $\bbC^{d_1}\ot\bbC^{d_1},\ldots,
 \bbC^{d_n}\ot\bbC^{d_n}$ (discussed in Section \ref{sec:optimal}).
\item The average overlap of $\ket{\psi}$ with a {\em random} product
  state, and a quantum variant of the Gowers uniformity norm \cite{gowers98} (discussed in Appendix \ref{sec:average}).
\item  The
  average purity of $\ket{\psi}$ across a random partition of $[n]$
  into two subsets (also discussed in Appendix \ref{sec:average}).
\eit

%-------------------------------------------------------------------------------

\subsection{Implications for classical computer science}
The main result of our paper proposes a quantum solution to a quantum
problem.  Nevertheless, there are several implications of our product
test that may be of interest to classical computing.  Instead of
viewing our results as concerning entangled states of many systems,
they may be interpreted in terms of tensors with many indices.  These
tensors have been studied in the context of image
processing~\cite{aja09tensors}, the planted clique
problem~\cite{BV09}, constraint satisfaction
problems~\cite{delaVegaKKV05} and many other
settings~\cite{vanLoan09}.

In this language, our results in \secref{complex} imply that many
central tensor problems, such as the injective tensor norm (defined in
\secref{equivalent}), are hard to approximate even to within constant
factors.  On the positive side, our \thmref{hSep-equiv} (together with
the equivalences in \secref{equivalent}) implies that if a heuristic
or approximation algorithm existed to optimise over trilinear forms,
it could be extended with little loss of accuracy, to perform
optimisations over $k$-linear forms for general $k$.  These
connections have been further explored in \cite{BBHKSZ12}, which shows
that the $\ell_2\ra\ell_4$ norm of a matrix is hard to approximate,
and connects this problem to the small-set expansion problem.

\subsection{Related work}
Our paper addresses a central problem in multipartite entanglement, which is too vast a field to reasonably summarise here (one good recent survey is \cite{HHH-entanglement}). We therefore concentrate on reviewing work on quantum Merlin-Arthur games with multiple provers.
% The particular measure of entanglement most relevant to our work is the geometric measure of entanglement, about which little is known.
%  Some recent results on the geometric measure of entanglement are [0905.4822], for symmetric states, and [Ghar-Kempe], which establishes

The class $\QMA(k)$ was first introduced by Kobayashi, Matsumoto, and
Yamakami~\cite{kobayashi03}, who showed that amplification of the
soundness gap of $\QMA(2)$ protocols implies that $\QMA(2)=\QMA(k)$ (a
result proven independently in~\cite{Unentanglement}). Both these
papers also showed that it is possible to amplify {\em completeness}
to exponentially close to 1 (see \lemref{completeness} for a
restatement). Blier and Tapp showed~\cite{BT09} that graph
3-colourability can be decided using a $\QMA(2)$ protocol with
messages of length $O(\log n)$ qubits and soundness
$1-\Omega(1/n^6)$. This soundness gap was improved to
$\Omega(1/n^{3+\epsilon})$ by Beigi~\cite{beigi10}, and has recently
been improved again to $\Omega(1/(n \polylog n))$ by Le Gall, Nakagawa
and Nishimura~\cite{legall11}. By contrast, Aaronson et al.\ 
showed~\cite{Unentanglement} that $\sat$ can be solved by a $\QMA(k)$
protocol with {\em constant} soundness, at the expense of increasing
$k$ to $O(\sqrt{n} \polylog(n))$.  Finally, Liu, Christandl and
Verstraete have given a problem in $\QMA(2)$ which is not obviously in
$\QMA$~\cite{liu07a}.

Following the conference and arXiv versions of this work, there have
been several interesting developments related to $\QMA(k)$.
First, it has been shown by Brand\~{a}o, Christandl and Yard~\cite{brandao11,brandao10a} that $\QMA(k)$ protocols, for constant $k$, are no stronger than $\QMA$ protocols if the verifier's measurement is restricted to be LOCC (implementable via local operations and classical communication). One consequence of their work is that, if there existed an efficient LOCC product state test, $\QMA(k)=\QMA$. However, we show in Appendix \ref{sec:locc} that no such test can exist. In the same work, the authors give a subexponential-time algorithm for optimizing over the set of separable states~\cite{brandao10a}; an alternative algorithm for this task has been given by Shi and Wu~\cite{shi12}, who also prove that several special cases of $\QMA(2)$ protocols can be simulated in polynomial space.

On the other hand, Chen and Drucker~\cite{chen10} have improved on the result of \cite{Unentanglement} and given an LOCC $\QMA(k)$ protocol that verifies $\sat$ with constant soundness gap for $k=O(\sqrt{n}\poly\log(n))$. (In fact, their protocol fits in the more restrictive class known as $\BellQMA(k)$.) Chiesa and Forbes~\cite{chiesa11} recently gave a tight soundness analysis of this protocol, showing that the soundness gap increases smoothly with $k$.

%Gharibian, Sikora and Upadhyay~\cite{gharibian11} have shown that if
%the proofs in a $\QMA(k)$ protocol are restricted to be of logarithmic
%size, the resulting complexity class is just $\QCMA$, the class of
%problems whose solutions can be verified by a quantum verifier given a
%classical proof.
McKague has recently used one of our results (that
the verifier's ``yes'' measurement operator may be taken to be
separable) to prove that restricting the class $\QMA(2)$ to real
Hilbert space does not change its computational
power~\cite{mckague11}.  While it is natural to expect the real case
to behave similarly to the complex case, we note that even for states
with real coefficients the closest product state may be
complex~\cite{CKP00}.
Finally, another application of our results was found by \cite{CS11}, who have presented
the only known nontrivial $\QMA(2)$-complete problem: estimating the
minimum energy of a sparse Hamiltonian over all bipartite product states.

\subsection{Organisation}

The remainder of this paper is organised as follows. In Section
\ref{sec:strategy}, we give an overview of the proofs of our main
results (details are in Appendices \ref{sec:depolarising} and
\ref{sec:prodtest}). In Section \ref{sec:QMA}, we
apply the product test to prove that $\QMA(k)=\QMA(2)$ for $k \ge 2$,
and we give some complexity-theoretic applications of this result in
Section \ref{sec:complex}, including an extensive discussion of
problems related to $\QMA(2)$. In Section \ref{sec:optimal} we argue that the
product test is essentially optimal, and in Section \ref{sec:depolar} we state our
results for the depolarising channel. We discuss the use of the product test
to test product unitaries in Section \ref{sec:unitaries}, and finish
with some open questions in Section \ref{sec:conc}.

%-------------------------------------------------------------------------------
\subsection{Notation}\label{sec:notation}
For a vector space $V$, define $B(V)$ to be the unit vectors in $V$, $L(V)$ to be the linear operators from $V$ to $V$, and $\mathcal{B}(V)$ to be the density operators on $V$. More concisely, let
$\cB(d)$ denote the set of $d\times d$ density matrices.
If $\ket\psi$ is a vector, let $\psi:=\proj\psi$.
Define the set of separable states on $\bbC^{d_A} \ot \bbC^{d_B}$ to be
\be
\Sep(d_A,d_B) := \conv\{\alpha
\ot \beta : \ket\alpha\in B(\bbC^{d_A}),\ket\beta\in B(\bbC^{d_B})\},
\ee 
where $\conv(S)$ denotes the convex closure of a set $S$. The swap operator on $\bbC^d \ot \bbC^d$ is denoted $\cF$, and is
defined to be $\sum_{i,j=1}^d \ket{i}\bra{j} \ot \ket j\bra i$.

For $\alpha\ge 1$, let $\|M\|_\alpha$ denote the Schatten $\alpha$-norm
of a matrix: $\tr (|M|^\alpha)^{1/\alpha}$. For a density matrix $\rho \in \mathcal{B}(\bbC^{d_1} \ot \cdots \ot \bbC^{d_n})$, and $S \subseteq \{1,\dots,n\}$, $\rho_S$ denotes the density matrix obtained by tracing out (discarding) the systems not in $S$. To avoid excessive parenthesization, we write $\rho_S^2 := (\rho_S)^2$ and $\tr \rho^2 := \tr (\rho^2)$.

\section{Overview of the proof of correctness}
\label{sec:strategy}

In this section, we sketch the proof of Theorem \ref{thm:prodtest}; a full proof is given in Appendices \ref{sec:depolarising} and \ref{sec:prodtest}.

We discuss here only the upper bound on $\Ptest$, since the lower
bound follows from continuity and the fact that product states pass
the test with probability 1.
First, we make precise the intuition that the product test is likely
to pass precisely when the average subsystem is close to pure.
\newcounter{test}\setcounter{test}{\value{thm}}
\begin{lem}
\label{lem:test}
Let $\Ptest(\rho,\sigma)$ denote the probability that the product test passes when applied to two mixed states $\rho,\sigma \in \mathcal{B}(\bbC^{d_1} \ot \cdots \ot \bbC^{d_n})$. Define $\Ptest(\rho) := \Ptest(\rho,\rho)$. Then
\[ \Ptest(\rho,\sigma) = \frac{1}{2^n} \sum_{S \subseteq [n]} \tr \rho_S \sigma_S, \]
and in particular
\[ \Ptest(\rho) = \frac{1}{2^n} \sum_{S \subseteq [n]} \tr \rho_S^2. \]
\end{lem}

%Our results are based on the following expression for the probability that an arbitrary state $\ket{\psi}$ passes the product test: $\Ptest(\proj{\psi}) = \frac{1}{2^n} \sum_{S \subseteq [n]} \tr \psi_S^2$. The connection with the depolarising channel follows from observing that the output purity of a state to which $n$ copies of the depolarising channel with arbitrary noise rate have been applied is given by a similar but more complicated sum, with weights on each term.

The proof itself is split into two parts, beginning with the case
where $\epsilon$ is low. 
We write $\ket{\psi} = \sqrt{1-\eps} \ket{0^n}
+ \sqrt{\eps} \ket{\phi}$ without loss of generality, for some product
state $\ket{0^n}$ and arbitrary state $\ket{\phi}$. This allows an
explicit expression for $\tr \psi_S^2$ in terms of $\epsilon$ and
$\ket{\phi}$ to be obtained.  While the marginals of $\ket\phi$ can be
complicated (and worse, we need to consider products of expressions of the form
$\tr_{\bar S}\ket 0 \bra\phi$), we can simplify things by only
considering the reductions in $\tr\psi_S^2$ that occur when $\ket 0$
is combined with a state orthogonal to $\ket 0$.    Thus, we do not
need a detailed picture of $\ket\phi$, but instead will merely split
it into a superposition of strings with different Hamming weight
(i.e. number of positions orthogonal to $\ket 0$).  Intuitively, the
contribution to $\bbE_{S\subseteq [n]} \tr\psi_S^2$ of a
piece of $\ket \phi$ with Hamming weight $k$ should be exponentially
small in $k$, since each position that differs from 0 leads to a
constant reduction in weight when we project onto the symmetric
subspace.    In order to
obtain a non-trivial bound from this expression, the final stage of
this part of the proof is to use the fact that $\ket{0^n}$ is the
closest product state to $\ket{\psi}$ to argue that $\ket{\phi}$
cannot have any amplitude on basis states of Hamming weight 0 or 1.
Ruling out basis states of Hamming weight 0 (i.e. $\ket{0^n}$) is
obvious, since otherwise $\eps$ would be smaller.  Less obvious is
that $\ket\phi$ cannot have any amplitude on Hamming weight-1 states,
but this too is contradicted by the fact that $\ket{0^n}$ has overlap
with $\ket\phi$ that is a local
maximum among product states, and nonzero amplitude on weight-1 states
would mean an infinitesimal local rotation could reduce $\eps$.  As a
result, we obtain a bound that is applicable when $\eps$ is small.

\newcounter{prodtestsmalleps}\setcounter{prodtestsmalleps}{\value{thm}}
\begin{thm}
\label{thm:prodtestsmalleps}
Given $\ket{\psi}\in \bbC^{d_1} \ot \cdots \ot \bbC^{d_n}$, let
\[ 1 - \eps = \max\{ |\braket{\psi}{\phi_1,\ldots,\phi_n}|^2 : \ket{\phi_i}\in\bbC^{d_i}, 1 \le i \le n\}.\]
Then $1 - 2\eps + \eps^2 \le \Ptest(\proj{\psi}) \le 1 - \eps + \eps^2 + \eps^{3/2}$.
\end{thm}

In the case where $\epsilon$ is high, this result does not yet give a useful
upper bound. In the second part of the proof, we derive a constant
bound on $\Ptest(\proj{\psi})$ based on considering $\ket{\psi}$ as a $k$-partite
state, for some $k<n$. $\Ptest(\proj{\psi})$ can be shown to be upper
bounded by the probability that  the test for being product across any
partition into $k$ parties passes. Informally speaking, if
$\ket{\psi}$ is far from product across the $n$ subsystems, we show
that one can find a partition such that the distance from the closest
product state (with respect to this partition) falls into the regime
where the first part of the proof works.

\newcounter{prodtestlargeeps}\setcounter{prodtestlargeeps}{\value{thm}}
\begin{thm}
\label{thm:prodtestlargeeps}
Given $\ket{\psi}\in \bbC^{d_1} \ot \cdots \ot \bbC^{d_n}$, let
\[ 1 - \eps = \max\{ |\braket{\psi}{\phi_1,\ldots,\phi_n}|^2 : \ket{\phi_i}\in\bbC^{d_i}, 1 \le i \le n\}.\]
Then, if $\eps \ge 11/32 > 0.343$, $\Ptest(\proj{\psi}) \le 501/512 < 0.979$.
\end{thm}

Between them, Theorems \ref{thm:prodtestsmalleps} and
\ref{thm:prodtestlargeeps} imply Theorem \ref{thm:prodtest}.  In fact,
we can say precisely that $\Ptest(\proj\psi) = 1-c(\psi)\eps$ for $\frac{11}{512} \leq c(\psi) \leq 2$.

One feature of our proof that can be generalised is the expectation
over $S\subseteq [n]$.  We effectively choose $S$ by flipping a fair
coin, but if we use a biased coin then this has an interesting
alternate interpretation in terms of the output purity of the
depolarising channel.  This yields a similar result, which is not only
that product states maximise the output purity (as was previously
known), but that any state which even approximately maximises the
output purity must be approximately product.  See \secref{depolar} for
a precise statement.  Since the correctness of the product test is a
special case of this more general theorem, we first prove the result
about depolarising channels in Appendix~\ref{sec:depolarising} and
then complete the details necessary for the product test in
Appendix~\ref{sec:prodtest}.

This completes the overview of the proof; we now discuss some applications of the product test.

%-------------------------------------------------------------------------------

\section{\texorpdfstring{$\QMA(2)$}{QMA(2)} vs.\ \texorpdfstring{$\QMA(k)$}{QMA(k)}}
\label{sec:QMA}

In this section, we apply the product test to a problem in quantum complexity theory: whether $k$ unentangled provers are stronger than 2 unentangled provers. This question can be formalised as whether the complexity classes $\QMA(k)$ and $\QMA(2)$ are equal \cite{kobayashi03,Unentanglement}. These classes are defined as follows.

\begin{dfn}
A language $L$ is in $\QMA(k)_{s,c}$ if there exists a polynomial-time quantum algorithm $\mathcal{A}$ such that, for all inputs $x \in \{0,1\}^n$:

\begin{enumerate}
\item {\bf Completeness:} If $x \in L$, there exist $k$ witnesses $\ket{\psi_1},\dots,\ket{\psi_k}$, each a state of $\poly(n)$ qubits, such that $\mathcal{A}$ outputs ``accept'' with probability at least $c$ on input $\ket{x}\ket{\psi_1}\dots\ket{\psi_k}$.

\item {\bf Soundness:} If $x \notin L$, then $\mathcal{A}$ outputs ``accept'' with probability at most $s$ on input $\ket{x}\ket{\psi_1}\dots\ket{\psi_k}$, for all states $\ket{\psi_1},\dots,\ket{\psi_k}$.
\end{enumerate}

We use $\QMA(k)$ as shorthand for $\QMA(k)_{1/3,2/3}$, and $\QMA$ as shorthand for $\QMA(1)$. We always assume $1 \le k \le \poly(n)$.

We also define $\QMA_m(k)_{s,c}$ to indicate that $\ket{\psi_1},\ldots,\ket{\psi_k}$ each involve $m$ qubits, where $m$ may be a function of $n$ other than $\poly(n)$.
\end{dfn}

Two of the major open problems related to $\QMA(k)_{s,c}$ are to
determine how the size of the complexity class depends on $k$ and on
$s,c$.  It has been conjectured for some years \cite{kobayashi03,Unentanglement} that in
fact $\QMA(k)=\QMA(2)$ for $2\leq	k\leq \poly(n)$, and that the soundness and completeness can be amplified by parallel repetition in a way similar to $\BPP$, $\BQP$, $\MA$, $\QMA$
and other complexity classes with bounded error.  In fact, these conjectures are related: $2k$ independent provers can simulate $k$ independent realisations of a $\QMA(2)$ protocol in order to amplify the soundness-completeness gap, and conversely,  \cite{kobayashi03,Unentanglement} proved that $\QMA(2)$ amplification implies that $\QMA(2)=\QMA(\poly)$.
  In this section, we will fully resolve these conjectures, proving that $\QMA(2)=\QMA(\poly)$ and that $\QMA(k)$ can have its soundness and completeness amplified by a suitable protocol.

The most direct way of putting $\QMA(k)$ inside $\QMA(2)$ is to ask two provers to each send the $k$ unentangled proofs that correspond to a $\QMA(k)$ protocol.  If $k=\poly(n)$, then each prover is still sending only polynomially many qubits.  Then the product test can be used to verify that the states sent were indeed product states and can be used as valid inputs to a $\QMA(k)$ protocol.  The specific protocol is described in \protoref{qmak}.

\boxproto{$\QMA(k)$ to $\QMA(2)$}{
\label{proto:qmak}
The $\QMA(2)$ protocol proceeds as follows.
\begin{enumerate}
\item Each of the two Merlins sends $\ket{\psi}:=\ket{\psi_1}\ot\ldots\ot \ket{\psi_k}$ to Arthur.
\item Arthur performs each of the following tests with
  probability $1/2$.
\benum
\item Arthur runs the product test with the two states as input and
  accepts iff the test outputs ``product.''
\item Arthur randomly chooses one of the states from the two Merlins,
  runs the algorithm $\mathcal{A}$ on that state, and outputs the result.
\eenum
\end{enumerate}
}

First observe that for YES instances (instances in the language), $k$
Merlins can achieve success probability $\geq c$, so by sending two
copies of this optimal state, two Merlins can achieve completeness
$\geq \frac{1+c}{2}\geq c$ in this modified protocol.

Now consider NO instances. 
Assume for now that the two Merlins always send the same state.  Then according to \thmref{prodtest}, if the Merlins send states that are far from product, they are likely to fail the product test, whereas basic continuity arguments can show that if they send states that are nearly product then the success probability will not be much larger than the soundness of the original protocol.  Thus, the soundness does not become too much worse.
These ideas (with a detailed proof in \appref{protoproof}) establish
\def\lemqmaksim{
For any $m$, $k$, $0\leq s<c\leq 1$, 
$$\QMA_m(k)_{s,c} \subseteq \QMA_{km}(2)_{s',c'}$$
where $c'=\frac{1+c}{2}$ and $s' = 1- \frac{(1-s)^2}{100}$.
}
\begin{lem}\label{lem:qmak-sim}
\lemqmaksim
\end{lem}

This is already strong enough to achieve amplification up to constant
soundness.  However, \protoref{qmak} has a salutary side effect that
will allow us to achieve stronger amplification.  To see this, we will
first introduce a further generalisation of the $\QMA(k)$ family.  Let
$\bbM$ be a set of Hermitian operators $M$ satisfying $0\leq M\leq I$.
Each $M\in \bbM$ defines a binary measurement with $M$ corresponding
to the ``accept'' outcome and $I-M$ corresponding to the ``reject''
outcome.  Variants of $\QMA(2)$ have been considered in which $M$ is
not only restricted to be efficiently implementable on a quantum
computer, but also with the further restriction that it belongs to
some set $\bbM$.  We will consider standard classes of measurements
such as BELL, LOCC, SEP, ALL, etc., whose definitions we include in
Appendix~\ref{sec:measure-defs}.
 
Formally, define $\QMA_m^{\bbM}(k)_{s,c}$ to be the class
$\QMA_m(k)_{s,c}$ with Arthur restricted to performing measurements
from $\bbM$ in addition to being restricted to quantum polynomial time.  For example, 
%if $\bbM=\text{PROD}$ is the set of product measurements followed by post-processing, then
 $\QMA^{\text{BELL}}(k)$ has been introduced under the name
 $\BellQMA(k)$ and proven equal to $\QMA$ (for constant $k$) by
 Brand\~{a}o \cite{brandao08,Unentanglement}.  Our paper will focus on
 the case that $\bbM=\SEP$, the class of measurements such that $M$ is
 a separable operator.   Observe that $M\in\SEP$ does not imply that $I-M$ is
 separable.

Armed with the definition of $\QMA^{\SEP}$, we can now see that \protoref{qmak} produces a protocol that is not only in $\QMA(2)$, but also $\QMA^{\SEP}(2)$.  More formally, we can strengthen \lemref{qmak-sim} to:
\begin{lem}\label{lem:qmak-sim-sep}
For any $m$, $k$, $0\leq s<c\leq 1$, 
$$\QMA_m(k)_{s,c} \subseteq \QMA_{km}^{\SEP}(2)_{s',c'}$$
where $c'=\frac{1+c}{2}$ and $s' = 1- \frac{(1-s)^2}{100}$.
\end{lem}

\begin{proof}
We again use \protoref{qmak}.  By \lemref{qmak-sim}, we know that this protocol has completeness $c' = (1+c)/2$, and has soundness $s'=1-\Omega((1-s)^2)$.  It remains only to argue that the ``accept'' measurement outcome is a separable operator.

Suppose the first Merlin sends systems $A_1,\ldots,A_k$ and the second Merlin sends systems $B_1,\ldots,B_k$.
The ``accept'' outcome of the product test corresponds to the tensor
product of projectors onto the symmetric subspaces of $A_1B_1, A_2B_2,
\ldots, A_kB_k$.  Since the symmetric subspace is spanned by vectors
of the form $\ket\psi^{\ot 2}$, it follows that the projector onto each
symmetric subspace is separable across the $A:B$ cut, and in turn that
their tensor product is as well. The other test in the protocol is to
simply apply a measurement either entirely on $A_1,\ldots,A_k$ or
entirely on $B_1,\ldots,B_k$, which is automatically separable.
Finally, performing a probabilistic mixture of separable measurements
creates a composite measurement which is also separable. 
\end{proof}

The advantage of $\QMA^{\SEP}(k)$ is that it removes the chief difficulty with $\QMA(k)$ amplification, which is that conditioning on measurement outcomes can induce entanglement between systems we have not yet measured.  This phenomenon is known as entanglement swapping~\cite{ent-swap}.  However, if we condition on the outcome of a measurement being $M$, for some $M\in \SEP$, then no entanglement will be produced in the unmeasured states.  As a result, cheating provers cannot gain any advantage by sending entangled proofs, and we obtain the following lemma.

\def\lemqmasepampl{
For any $\ell\geq 1$,
$$\QMA_m^{\SEP}(k)_{s,c} \subseteq \QMA_{\ell m}^{\SEP}(k)_{s^\ell,c^\ell}.$$
}
\begin{lem}\label{lem:qma-sep-ampl}
\lemqmasepampl
\end{lem}

The idea is to simply repeat the original protocol $\ell$ times in
parallel and to accept iff each subprotocol accepts.  Since we are
considering $\QMA^{\SEP}$ protocols, obtaining an ``accept'' outcome on one proof will not induce any entanglement on the remaining proofs.  We give a detailed proof of \lemref{qma-sep-ampl} in \appref{protoproof}.

From \lemref{qmak-sim-sep} and \lemref{qma-sep-ampl}, we can almost conclude that strong amplification is possible.  Indeed, when we start with protocols with perfect completeness, we can apply \protoref{qmak}, repeat $p(n)$ times, and reduce the soundness from $s$ to $s^{O(p(n))}$.  For the case of $c<1$, we need one additional argument to keep the completeness from being reduced too much at the same time.  Here we will use a method for completeness amplification proved  in both \cite[Lemma 5]{kobayashi03} and \cite[Lemma 6]{Unentanglement}.
\def\lemcompleteness{
For any $\ell\geq 1$,
$$\QMA_m(k)_{s,c} \subseteq \QMA_{\ell m}(k)_{1-\frac{c-s}{3}, 1-\exp(-\frac{\ell(c-s)^2}{2})}.$$
}
\begin{lem}[\cite{kobayashi03,Unentanglement}]\label{lem:completeness}
\lemcompleteness
\end{lem}
Our amplification procedure for general $c<1$ is then to
\benum
\item Use \lemref{completeness} to bring the completeness exponentially close to 1.
\item Use \lemref{qmak-sim-sep} to convert a general $\QMA(k)$ protocol to a $\QMA^{\SEP}(2)$ protocol.
\item Repeat the protocol polynomially many times to make the soundness exponentially small.
\eenum

This procedure then achieves
\def\thmqma2k{
\benum \item
If $s\leq 1-1/\poly(n)$, $k=\poly(n)$ and $p(n)$ is an arbitrary polynomial, then
$\QMA(k)_{s,1} = \QMA^{\SEP}(2)_{\exp(-p(n)), 1}$.
\item
If $c-s \geq 1/\poly(n)$, $c<1$, $k=\poly(n)$ and $p(n)$ is an arbitrary polynomial, then
$\QMA(k)_{s,c} = \QMA^{\SEP}(2)_{\exp(-p(n)), 1-\exp(-p(n))}$.
\eenum
}
\begin{thm}\label{thm:qma2k}
\thmqma2k
\end{thm}

We prove correctness of Protocol \ref{proto:qmak} and the rest of
\thmref{qma2k} in Appendix \ref{sec:protoproof}.  There are obvious variants of \thmref{qma2k} to cover the case of limited message size, whose statements we leave implicit.

%-------------------------------------------------------------------------------

\subsection{\texorpdfstring{$\QMA(2)$}{QMA(2)} and \texorpdfstring{$h_{\Sep}$}{h_Sep}}
\label{sec:qma2hsep}

The complexity of $\QMA_m(2)$ stems both from the complexity of
producing the measurement made by the verifier, and of maximising its
acceptance probability over product states.  To understand the
complexity of this second step, we define the support function of the
separable states to be
\be h_{\Sep(d,d)}(M) := \max\{\tr M\rho : \rho\in\Sep(d,d)\}
= \max\{\tr M(\alpha\ot \beta): \ket\alpha,\ket\beta\in B(\bbC^d)\},
\ee
for any $M\in L(\bbC^d\ot\bbC^d)$.  
Calculating $h_{\Sep}$ up to $1/\poly(d)$ accuracy was proven to be
$\NP$-hard by Gurvits~\cite{gurvits03}.  One of our central results
will be a weaker hardness result for the problem of estimating
$h_{\Sep}$ to {\em constant} additive error; see \thmref{hSep-equiv} below.

Thus, $\QMA_m(2)_{s,c}$ is the class of languages that can be decided by
determining whether $h_{\Sep(2^m,2^m)}(M)$ is $\geq c$ or $\leq s$,
where $M$ is a measurement operator that
can be constructed in polynomial time on a quantum computer.  It is
instructive to compare the case of $\QMA_m(1)$, where the problem can
be thought of as computing the largest eigenvalue of a $2^m\times 2^m$
matrix.  There the hardness comes from the fact that the matrix is
implicitly specified by a polynomial-size quantum circuit.  By
contrast, in the case of $\QMA(2)$, there is no known $\poly(d)$-time
algorithm to compute $h_{\Sep(d,d)}(M)$.
As a result, $\QMA_{\log}(1)=\BQP$~\cite{marriott05}, but $\QMA_{\log}(2)$ is not known
to be in $\BQP$ (since we do not know how to search over unentangled
pairs of $\log(n)$-qubit states in quantum polynomial time) or $\NP$
(since the measurement can depend on a general quantum poly-time
algorithm).   The weakest class that we know contains
$\QMA_{\log}(2)$ is
$\NP^{\BQP}$, by using the $\BQP$ oracle to obtain an
explicit description of $M$.  This can be achieved up to error $\eps$
by running the verifier's circuit $\poly(2^m,1/\eps)$ times and
performing tomography.  We therefore obtain that
$\QMA_m(2)_{s,c} \subseteq \NTIME(\poly(2^m,n,1/(c-s)))^{\BQP}$.  In particular,
 $\QMA(2)\subseteq \NEXP$.
 Unfortunately this cannot be scaled down to place $\QMA_{\log}(2)$
 in $\NP$.  This is because the verifier in a $\QMA_{\log}(2)$
 protocol still can perform a poly-time quantum computation.  Thus,
 we only have that $\QMA_{\log}(2)\subseteq\NP^{\BQP}$.
 
Via the connection between $\QMA_m(2)$ and $h_{\Sep(2^m,2^m)}$, all of the results in this section can be stated in terms of $h_{\Sep}$. In particular, we have the following variants of Lemma \ref{lem:qma-sep-ampl} and Theorem \ref{thm:qma2k}.
 
\begin{lem}\label{lem:hsep-sep-ampl}
Let $M \in SEP$ be a $d^2\times d^2$-dimensional separable Hermitian matrix satisfying
$0\leq M\leq I$. Then $h_{\Sep(d^k,d^k)}(M^{\otimes k}) = h_{\Sep}(M)^k$.
\end{lem} 
 
\begin{thm}\label{thm:hSep-equiv}
Let $M$ be a $d^2\times d^2$-dimensional Hermitian matrix satisfying
$0\leq M\leq I$.  Assume that we are promised that $h_{\Sep(d,d)}(M)$
is either $\geq c$ or $\leq s$ for $1\geq c > s > 0$.  Call these
two cases ``Y'' and ``N.''   Choose $1\geq c'>s'>0$ such that $c'=1$ if
and only if $c=1$.
Then there exists a matrix $M'$ of size
$d^k$ such that 
\be h_{\Sep(d^k,d^k)}(M')  \begin{cases}
\geq c' & \text{in case Y} \\
\leq s' & \text{in case N}
\end{cases}\ee
If $c=1$, then $k = O((1-s)^{-2}\log(1/s'))$, and if $c<1$, then $k =
O((c-s)^{-3}\log(1/(1-c'))\log^2(1/s'))$.  Additionally $M'\in \SEP$,
and $M'$ can be constructed efficiently from $M$, even by a classical
log-space transducer.
\end{thm}

%-------------------------------------------------------------------------------

\section{Complexity-theoretic implications}
\label{sec:complex}
\subsection{Evidence for the hardness of \texorpdfstring{$\QMA_{\log}(2)$}{QMA\_log(2)}}
A key application of \thmref{qma2k} is to the protocol of
Ref.~\cite{Unentanglement} that puts \sat\ on $n$ clauses inside the complexity class
$\QMA_{\log(n)}(\sqrt{n}\poly\log(n))_{1-\Omega(1),1}$.
Applying \thmref{qma2k} lets us simulate this using two
provers with perfect completeness and arbitrary soundness, so
that we obtain 
\begin{cor}\label{cor:3-sat}
Let $\ell:\bbN \rightarrow \bbN$ be polynomially bounded. Then
$$\sat \in \QMA_{\ell(n)\sqrt{n}\poly\log(n)}(2)_{2^{-\ell(n)},1}.$$
In other words, there is a protocol for  \sat\ instances with $n$
clauses
that uses two provers,
$\ell(n)\sqrt{n}\poly\log(n)$-qubit proofs and has 
perfect completeness and soundness $2^{-\ell(n)}$. 
\end{cor}

Therefore, making assumptions about the hardness of $\sat$ allows us to prove hardness results
for the complexity class $\QMA_{\log}(2)$, and stronger assumptions naturally
imply stronger hardness results. We formalise this correspondence
as the following corollary.

\begin{cor}\label{cor:3-sat-log}
The following implications hold.
\begin{enumerate}[(i)]
\item If $\sat$ on $n$ clauses is not in $\DTIME(\exp(o(n)))$, then for arbitrary constant $\epsilon > 0$
\[ \QMA_{\log(d)}(2)_{\frac{1}{2},1} \nsubseteq \DTIME(d^{\log^{1-\epsilon} d}). \]

\item If $\sat$ on $n$ clauses is not in $\DTIME(\exp(o(n)))$, then
\[ \QMA_{\log(d)}(2)_{2^{-\sqrt{\log d}/\polylog(\log d)},1} \nsubseteq \DTIME(\poly(d)). \]

\item If $\sat$ on $n$ clauses is not in $\DTIME(\exp(\sqrt{n}\polylog(n)))$, then
\[ \QMA_{\log(d)}(2)_{\frac{1}{2},1} \nsubseteq \DTIME(\poly(d)). \]
\end{enumerate}

More generally, assume that for some functions $\ell,m:\bbN\rightarrow \bbN$, $\sat$ on $n$ clauses is not contained in $\DTIME(m(\exp(\ell(n) \sqrt{n} \polylog(n))))$. Then, defining $d = 2^{\ell(n)\sqrt{n}\polylog(n)}$,
\[ \QMA_{\log(d)}(2)_{2^{-\ell(n)},1} \nsubseteq \DTIME(m(d)). \]
\end{cor}

\begin{comment}
\begin{cor}\label{cor:3-sat-log}
Assume that, for some functions $\ell,m:\bbN\rightarrow \bbN$, $\sat$ on $n$ clauses is not contained in $\DTIME(m(\exp(\ell(n) \sqrt{n} \polylog(n))))$. Then, defining $d = 2^{\ell(n)\sqrt{n}\polylog(n)}$,
%
\[ \QMA_{\log(d)}(2)_{2^{-\ell(n)},1} \nsubseteq \DTIME(m(d)). \]
%
 In particular, we have the following implications.
%
\begin{enumerate}[(i)]
%
\item Taking $\ell(n)=1$, $m(d)=\poly(d)$,
%
\[ \QMA_{\log(d)}(2)_{\frac{1}{2},1} \nsubseteq \DTIME(\poly(d)) \]
%
under the assumption that $\sat \not \in \DTIME(\exp(\sqrt{n}\polylog(n)))$.

\item Taking $\ell(n)=1$, $m(d) = \exp(\log^{2-\epsilon} d)$,
%
\[ \QMA_{\log(d)}(2)_{\frac{1}{2},1} \nsubseteq \DTIME(d^{\log^{1-\epsilon} d}) \]
%
for arbitrary constant $\epsilon > 0$, under the assumption that $\sat \not \in \DTIME(\exp(o(n)))$.

\item Taking $\ell(n)=\sqrt{n}/\polylog(n)$, $m(d) = \poly(d)$,
%
\[ \QMA_{\log(d)}(2)_{2^{-\sqrt{n}/\polylog(n)},1} \nsubseteq \DTIME(\poly(d)) \]
%
under the assumption that $\sat \not \in \DTIME(\exp(o(n)))$.
\end{enumerate}
\end{cor}
\end{comment}

Note that the assumptions on the hardness of $\sat$ made in the first two cases are essentially equivalent to the (not implausible) {\em Exponential Time Hypothesis} of Impagliazzo and Paturi~\cite{impagliazzo01}, which states that $\sat\not\in\DTIME(\exp(\ell(n)))$ for any $\ell(n) = o(n)$. 

\subsection{\texorpdfstring{$\QMA_{\log}(2)$}{QMA\_log(2)} equivalences and reductions}
\label{sec:equivalent}

We now discuss a number of problems for which Corollary \ref{cor:3-sat-log} allows us to prove hardness results. As described in Section \ref{sec:qma2hsep}, $\QMA_{\log(d)}(2)$ is intimately connected to $h_{\Sep(d,d)}$.
Here we focus solely on the hardness of
estimating $h_{\Sep(d,d)}(M)$ when $0\leq M\leq I$ is given explicitly
as input.  In other words, we will examine the part of the hardness of
$\QMA_{\log}(2)$ that does {\em not} come from having access to a
poly-time quantum computation.

One definition we will repeatedly use is that of the {\em weak
  membership problem}.  If $K$ is a convex set, $\eps>0$ and $d$ is a
metric, then $\WMEM_\eps^{(d)}(K)$ denotes the following task: given
an input $x$, determine whether $x\in K$ or $d(x,K)\geq \eps$, given
the promise that one of these conditions holds.  Here
$d(x,K):=\inf_{y\in K}d(x,y)$.  The reason for the $\eps$ is because
the complexity of the problem can depend on the required precision,
just as the size of $\QMA(k)_{s,c}$ depends on how close $s$ and $c$
are.  See \cite{GLS} for more background and equivalent formulations
of the weak membership problem for convex sets.  In many cases, $d$
will be the trace norm distance; in this case, we will simply write
$\WMEM_\eps(K)$ for the weak membership problem.  We also define
$B_d(K,\eps) := \{x : d(x,K)\leq\eps\}$, and define the Hausdorff
distance between (not necessarily convex) sets $K,L$ to be
$d_H(K,L):=\max\{\sup_{x\in K}d(x,L), \sup_{x\in L}d(x,K)\}$, or
equivalently, $\inf\{\eps\geq 0 : X\subseteq B_d(Y,\eps) \text{ and
}Y\subseteq B_d(X,\eps)\}$.

The following equivalences and reductions are a combination of known
results (\cite{werner02,gurvits03,EHGC04,Mat05,fannes06,liu07,BV09,gharibian11,BBHKSZ12}, and some unpublished and/or folklore)
and consequences of our 
main theorems.  Even though many of the reductions are
straightforward, we are not aware of any similar list in the
literature, despite many of the quantities being discussed
individually.

\noindent{\bf Equivalent problems}
\begin{enumerate}
\item\label{it:hsep}
 Given $M$ with $0\leq M\leq I$, determine whether
\be
h_{\Sep}(M):= \max_{\rho\in\Sep(d,d)} \tr M\rho
 = \max_{\ket\alpha,\ket\beta\in B(\bbC^d)} \tr M(\alpha\ot \beta)
\ee
is $\geq c$ or $\leq s$.
As discussed above, this represents the acceptance probability of a
$\QMA(2)$ protocol when the measurement is fixed and the provers use
an optimal strategy.  

This is our reference problem, and we will
compare the problems below to this one.  However, we observe that this
problem is equivalent (up to a polynomial change of dimension described below) to
versions with different choices of $c$ and $s$ as long as $0<s<c<1$
are constants independent of dimension.
\item\label{it:hprod-sym}
Define ${\rm ProdSym}(d):=\conv\{\psi\ot\psi: \ket\psi\in B(\bbC^d)\}$.
Given $M$ with $0\leq M\leq I$, determine whether $h_{{\rm
    ProdSym}(d)}(M)$ is $\geq c/4$ or $\leq s/4$.
\item\label{it:hsep-sym}
Define ${\rm SepSym}(d):=\conv\{\rho\ot\rho: \rho\in \cB(d)\}$.
Given $M$ with $0\leq M\leq I$, determine whether $h_{{\rm
    SepSym}(d)}(M)$ is $\geq c/4$ or $\leq s/4$.
\item \label{it:EW}
The set EW:=EW$(d,d)$ of {\em entanglement witnesses}~\cite{terhal00} is the dual cone of
  Sep, meaning that
\be {\rm EW}(d,d) = \{W : \tr W\rho \geq 0 , \forall \rho\in\Sep(d,d)\}.
\label{eq:EW}\ee
Given $M$, determine whether $\min \{ \|M+W\|
: W\in {\rm EW}(d,d)\}$   is $\geq c$ or $\leq s$.
\item \label{it:1inf-norm}
For a quantum channel $\cN$, determine whether the superoperator $1\ra\infty$ norm
$\|\cN\|_{1\ra \infty}$ is $\geq c$ or $\leq s$, where $\|\cN\|_{1\ra
  \infty} := \max_{\rho}\frac{\|\cN(\rho)\|_\infty}{\|\rho\|_1}$.
\item \label{it:12-norm} For a quantum channel $\cN$, determine
  whether the superoperator $1\ra 2$ norm $\|\cN\|_{1\ra 2}$ is $\geq
  4c-3$ or $\leq 4s-3$, where $\|\cN\|_{1\ra 2} :=
  \max_{\rho}\frac{\|\cN(\rho)\|_2}{\|\rho\|_1}$.  (This equivalence
  is only nontrivial for some values of $c,s$.)
\item \label{it:Sinfmin}
Given $\cN$, determine whether the minimum output R\`enyi entropy $S_\infty^{\min}(\cN)$ is $\geq \log(1/s)$ or $\leq \log(1/c)$.
Here $S_\infty^{\min}:=\min_\rho S_\infty(\rho)$, where $S_\infty(\sigma):=-\log \|\sigma\|_\infty$.
\item \label{it:S2min}
Given $\cN$, determine whether the minimum output R\`enyi entropy
$S_2^{\min}(\cN)$ is $\geq \log(2/\sqrt s)$ or $\leq \log(2/\sqrt c)$.
Here $S_2^{\min}:=\min_\rho S_2(\rho)$, where $S_2(\sigma):=-\log \|\sigma\|_2$.
\item \label{it:inj}
Given a 3-index tensor $T\in \bbC^d\ot\bbC^d\ot\bbC^d$, determine whether the injective tensor norm $\|T\|_{\rm inj}$ is $\geq \sqrt c$ or $\leq \sqrt s$.  The injective tensor norm is defined here to be
\be \|T\|_{\rm inj} = \max_{x,y,z\in B(\bbC^d)} |\bra{T}\cdot\ket x
\ot \ket y \ot \ket z|\ee
and $T$ should have the property that for some choice of indices it
can be interpreted as a linear map from $\bbC^d \ra \bbC^{d^2}$ with
operator norm $\leq 1$.
\item \label{it:2toinf}
 Given a linear map $T$ from $\bbC^d$ to $L(\bbC^d)$ with operator norm $\leq 1$, determine whether $\|T\|_{\ell_2\ra S_\infty}$ is $\geq \sqrt c$ or $\leq \sqrt s$.  Here $\ell_2$ is the usual vector 2-norm and we use $S_\infty$ to emphasise that the output norm is the Schatten $\infty$-norm for operators.
\item \label{it:geom}
Given a pure state $\ket\psi\in B(\bbC^d\ot\bbC^d\ot\bbC^d)$, define the geometric measure of entanglement
\be E_{\rm geom}(\ket\psi) = -\log \max_{x,y,z\in B(\bbC^d)} |\bra{\psi}\cdot\ket x \ot \ket y \ot \ket z|^2 .\ee
Then determine whether $E_{\rm geom}(\ket\psi)$ is $\leq \log(d/c)$ or
$\geq \log(d/s)$, for $\ket{\psi}^{ABC}$ satisfying $\psi_A = I/d$.
\item \label{it:min-ent}
Given a subspace $V \subseteq \bbC^d \ot \bbC^d$, define the minimum entanglement of $V$ to be
\be \nu_\infty(V) := \min_{\ket\psi\in B(V)} \|\tr_1 \proj\psi\|_\infty,\ee
where $\tr_1$ means the partial trace over the first subsystem.  Then determine whether $\nu_\infty(V)$ is $\geq c$ or $\leq s)$. 
\item \label{it:mean-field}\label{it:last-equiv}
Given a Hermitian $K\in L(\bbC^d \ot \bbC^d)$ with $0\leq K\leq I$, define the mean-field Hamiltonian $H_n\in L((\bbC^d)^{\ot n})$ by
\be H_n := \frac{1}{n(n-1)} \sum_{1 \leq i\neq j \leq n} K^{(i,j)} ,\ee
where $K^{(i,j)}$ indicates the operator with $K$ acting on systems $i,j$ and identity matrices elsewhere.  Let $\lambda_{\min}(H_n)$ denote the smallest eigenvalue of $H_n$.  Then determine whether $\lim_{n\ra \infty}\lambda_{\min}(H_n)$ is $\geq 1-s/2$ or $\leq 1-c/2$.
\suspend{enumerate}

The following problems can be reduced to and from estimating $h_{\Sep}$, but unlike the above problems, the reductions no longer preserve the same completeness and soundness.

\noindent{\bf Approximately equivalent problems}\vspace{-4pt}
\resume{enumerate}
\item \label{it:weak-mem}
Separability testing: given a state $\rho$ and a promise that it
  is either separable or a constant distance away from separable in the trace norm,
  determine which is the case.  In other words, solve
  $\WMEM_\eps(\Sep(d,d))$ for some $\eps>0$.
\item \label{it:weak-EW}
Weak membership for entanglement witnesses (defined in \eq{EW}), with
distance defined in operator norm; i.e. $\WMEM_\eps^\infty({\rm EW})$.
\item \label{it:more-parties}
Injective tensor norm for $k$-partite states with $k\geq 4$, geometric
measure of entanglement for $k$-partite states with $k\geq 4$, mean
field for interactions that are $k$-local for $k\geq 3$ and 
$h_{\Sep(d,d,d,\ldots)}$ and weak membership in $\Sep(d,d,d,\ldots)$
for more systems. 
\item \label{it:2-to-4}
Estimating the $\ell_2\ra \ell_4$ norm of a matrix, defined as $\|A\|_{\ell_2\ra\ell_4} := \sup_{x\neq 0} \|Ax\|_{\ell_4} / \|x\|_{\ell_2}$, where $\|x\|_{\ell p} := (\sum_{i=1}^d |x_i|^p)^{1/p}$.
\suspend{enumerate}

The following problems are at least as easy as $h_{\Sep}$, meaning that they can be reduced to estimating $h_{\Sep}$.  We will discuss below the specific parameters of the reductions.

\noindent{\bf Easier problems}\vspace{-4pt}
\resume{enumerate}
\item \label{it:3-sat}
Deciding 3-SAT$_{\log^2}$, which is defined to be the class of 3-SAT instances with $\log^2(d)$ variables and $O(\log^2(d))$ clauses.
By \corref{3-sat}, this reduces to $\QMA_{\log}(2)_{1/2, 1}$.
\item \label{it:clique}
The {\em planted clique problem} is to distinguish a $G_{n,1/2}$ graph (i.e. an undirected graph with $n$ vertices in which each edge is present with i.i.d.\ probability $1/2$) from the union of a $G_{n,1/2}$ graph and a random clique of size $n^{\beta}$.  For certain values of $\beta$, as we discuss below, this problem is known to reduce to estimating injective tensor norms.
\suspend{enumerate}

The following problems are at least as hard as estimating $h_{\Sep}$, meaning that $h_{\Sep}$ can be reduced to them,  in the special case when $c=1$.

\noindent{\bf Harder problems (when $c=1$)}\vspace{-4pt}
\resume{enumerate}
\item \label{it:Smin-alpha}
Given a channel $\cN$, determine whether $S_\alpha^{\min}(\cN)=0$ or is $\geq \log(1/s)$.
The minimum output R\'enyi $\alpha$-entropy of $\cN$ is defined to be
$S_\alpha^{\min}(\cN) = \min_\rho S_\alpha(\cN(\rho))$, where
$S_\alpha(\sigma) = \frac{1}{1-\alpha}\log\tr\sigma^\alpha$.
\item \label{it:Smin-reg}
Determine whether the {\em regularised} minimum output R\`enyi entropy $S_\alpha^{R, \min}(\cN)$ is 0 or $\geq \log(1/s)$.  Here $S_\alpha^{R, \min}(\cN)=\lim_{n\ra\infty} \frac{1}{n}S_\alpha^{\min}(\cN^{\ot n})$.
\end{enumerate}

Before explaining the connections between these problems, we note that
\corref{3-sat-log} can 
be restated in terms of $h_{\Sep}$, and thus also in terms of any of
the equivalent or harder problems. 

\begin{cor}\label{eq:equiv-hard}
Tasks \ref{it:hsep}-\ref{it:last-equiv} and
\ref{it:Smin-alpha}-\ref{it:Smin-reg} cannot be completed in time
$\poly(d)$ for any constants 
$0<s<c<1$ unless $\sat\in\DTIME(\exp(\sqrt{n}\poly\log(n)))$.
\end{cor}

\noindent{\bf Explanations}

\begin{enumerate}%\addtocounter{enumi}{1}
\item {\em Changing $c$ and $s$ for $h_{\Sep}$:}  This claim follows
  from \thmref{hSep-equiv}.  Similarly difficult is the problem of producing
  an estimate $X$ such that $|X-h_{\Sep}(M)|\leq \eps$ for some
  $\eps>0$.

One subtlety is that the $c=1$ case is not known to be equivalent to
the $c<1$ case.  Soundness, on the other hand, is always nonzero,
since we always have $h_{\Sep(d,d)}(M) \geq \tr M / d^2$.

\item {\em Estimating $h_{\rm ProdSym}$:}    We show this has
  equivalent difficulty to estimating $h_{\rm Sep}$.  Initially assume
 that we have an algorithm for estimating $h_{\rm ProdSym}$, and given $M\in L(\bbC^d \ot
  \bbC^d)$, would like to compute $h_{\rm Sep}(M)$.  Then define $M' =
  \ket{01}\bra{01} \ot M $. 
%+ \ket{10}\bra{10} \ot \cF M \cF$, where we recall that $\cF$ is the
%swap operator. 
 $M'$ is $4d^2$-dimensional, and if $0\leq M\leq I$, then $0 \leq M'\leq I$.

To calculate $h_{\rm ProdSym}$, we can without loss of generality let
 \be \ket\psi = \sqrt{p_0}\ket
0\ket\alpha + \sqrt{p_1}\ket 1\ket\beta,\ee
 where $p_0+p_1=1$ and
$\ket\alpha,\ket\beta\in B(\bbC^d)$.  Then $\tr M'(\psi \ot \psi) =
p_0p_1 \tr M(\alpha \ot \beta)$.  This is maximised when
$p_0=p_1=1/2$.  Thus
$$h_{{\rm ProdSym}(2d)}(M') = \frac{1}{4} h_{\Sep(d,d)}(M).$$

Conversely, suppose we are given an arbitrary $M$ and the ability to
compute $h_{\Sep}(\cdot )$ and would like to
estimate $h_{\rm ProdSym}(M)$.  First we assume $\cF M \cF=M$.  This
can be done WLOG since $h_{\rm ProdSym}(M) = h_{\rm ProdSym}((M + \cF
M \cF)/2)$. Then 
define $M' = \frac{\Pi_{\rm sym}^{d,2}
  + M}{2}$.  Our desired equivalence will follow from the following claim:
\be h_{\Sep}(M') = \frac{1+h_{\rm ProdSym}(M)}{2}.
\label{eq:Sep-ProdSym}\ee
One direction is easy: if $h_{\rm ProdSym}(M) = \tr M(\psi \ot \psi)$
then $h_{\Sep}(M') \geq \tr M' (\psi \ot \psi) = (1+h_{\rm
  ProdSym}(M))/2$.
To upper-bound $h_{\Sep}(M') = \max_{\alpha,\beta}\tr M'(\alpha \ot
\beta)$, we define $\theta,\ket a, \ket b$ such that
\bas \ket \alpha & = \cos(\theta/2) \ket a + \sin(\theta/2)\ket b\\
 \ket \beta & = \cos(\theta/2) \ket a - \sin(\theta/2)\ket b.\eas
To compute $\tr M'(\alpha \ot \beta)$, first we see that $\tr \Pi_{\rm
  sym}^{d,2}(\alpha \ot \beta) = (1 + \tr \alpha\beta)/2 =
1-\sin^2(\theta)/2$.  Next, we expand
$$\ket\alpha\ket \beta = 
\cos^2(\theta/2)\ket{aa} +
\sin(\theta/2)\cos(\theta/2)(\ket{ba}-\ket{ab})
- \sin^2(\theta/2)\ket{bb}.$$
When we expand $\bra{\alpha,\beta} M \ket{\alpha,\beta}$, the symmetry
of $M$ means that terms such as $\bra{aa}M(\ket{ba}-\ket{ab})$ vanish,
and we are left with
\bas &\cos^4(\theta/2)\bra{aa}M\ket{aa}
 + \sin^4(\theta/2)\bra{bb}M\ket{bb}
- \sin^2(\theta/2)\cos^2(\theta/2) (\bra{aa}M\ket{bb} +
\bra{bb}M\ket{aa}) \\
&+ 2\sin^2(\theta/2)\cos^2(\theta/2) 
\frac{\bra{ba}-\bra{ab}}{\sqrt{2}} M
\frac{\ket{ba}-\ket{ab}}{\sqrt{2}}.
\eas
Since $\|M\|\leq 1$, and using the definition of $h_{\ProdSym}$, we
have
\beas \tr M'(\alpha\ot\beta) &\leq&
1 -\frac{\sin^2(\theta)}{2} +
\frac{(\sin^4(\theta/2)+\cos^4(\theta/2))h_{\ProdSym}(M)
 + \sin^2(\theta)}{2}\\
&\leq& \frac{1+h_{\ProdSym}(M)}{2}.
\eeas
Maximising over all unit vectors $\alpha,\beta$, this establishes
\eq{Sep-ProdSym}.  We remark that \lemref{qmak-sim} would also relate
$h_{\Sep}$ and $h_{\rm ProdSym}$ but not in this exact fashion.
\item {\em Estimating $h_{\rm SepSym}$:}  Given $M$, let $M' =
  M^{A_1B_1} \ot I^{A_2 B_2}$.  Then $h_{\rm ProdSym}(M') =
  h_{\SepSym}(M)$. 

For the converse, we use the same construction as $h_{\ProdSym}$.  Assume
 that we have an algorithm for estimating $h_{\rm SepSym}$, and given $M\in L(\bbC^d \ot
  \bbC^d)$, would like to compute $h_{\rm Sep}(M)$.  Again we define $M' =
  \ket{01}\bra{01} \ot M$.  Let $\rho$ achieve the maximum of $\tr
  M'(\rho \ot \rho)$, and expand $\rho = 
\ket 0 \bra 0 \ot \rho_{00} + 
\ket 0 \bra 1 \ot \rho_{01} + 
\ket 1 \bra 0 \ot \rho_{10} + 
\ket 1 \bra 1 \ot \rho_{11},$ for some $\rho_{ij}\in L(\bbC^d)$.  Then
$\tr M'(\rho \ot \rho) = \tr M(\rho_{00} \ot \rho_{11})$.  Since
$\rho_{00}, \rho_{11}$ are proportional to density matrices, and
$\tr\rho=\tr\rho_{00} + \tr\rho_{11}$, the rest of the analysis
proceeds identically to in the case of $h_{\ProdSym}$.

\item {\em Entanglement witnesses:} $h_{\Sep}(M)$ is a convex program
whose dual is given by the minimisation of $\|M+W\|$ over $W\in {\rm
  EW}$.  See \cite{gharibian11} for a discussion of this point. 

\item {\em Estimating $\|\cN\|_{1\ra \infty}$:} 
\begin{comment}
{\em Old Proof:}
The idea is seen most simply when $M$ is a projector.  In this case,
we can define $\cN$ to be the channel from $\supp M$ to $\cB(d)$
that is obtained by
embedding $\supp M$ into $\bbC^d \ot \bbC^d$ and then tracing
out the second system.  The operator $M$ is related to the Stinespring
embedding \cite{Stinespring} of $\cN$. Now $\|\cN\|_{1\ra \infty}$ is
the largest eigenvalue of the least mixed state in the output of
$\cN$, which is equivalently the least entangled state in $\supp M$,
which in turn is the maximum acceptance possibility of the protocol.
In the case where $M$ is not a projector, then the corresponding $\cN$ is now
trace non-increasing instead of trace preserving, but the same
argument still goes through.  
This connection has been known for some
time as folklore and has appeared before in Ref.~\cite{Mat05}. 
We include a formal proof in Appendix~\ref{sec:inftynorm}.
\end{comment}
This connection has been known for some
time as folklore and has appeared before in Ref.~\cite{Mat05} (which
cites a personal communication from Watrous).
Since the largest value of $\|\cN(\rho)\|_\infty$ occurs when $\rho$
is pure, finding it corresponds to optimising a trilinear form over
unit vectors~\cite{werner02}; i.e. is equivalent to the injective tensor norm problem
described in task~\ref{it:inj}.  More concretely, define $V_\cN : \bbC^d
\ra \bbC^d \ot \bbC^d$ to be the isometric extension of $\cN$, so that
$\tr_E V_\cN \rho V_\cN^\dag = N(\rho)$.  Then 
\be \|\cN\|_{1\ra\infty} = \max_{\alpha,\beta,\gamma\in B(\bbC^d)}
\L|(\bra\beta\ot \bra\gamma)V_\cN\ket\alpha\R|^2.\ee
This expression equals $\|T\|_{\rm inj}^2$ (see task~\ref{it:inj}) for
$\ket T = \sum_{i=1}^d \ket i \ot V_\cN \ket i$.

\item {\em Estimating $\|\cN\|_{1\ra 2}$:} 
Define $M := (\cN^\dag \ot \cN^\dag)(\frac{I+\cF}{2})$. Then $h_{\ProdSym}(M) =
\max \{(1+\tr (\cN(\psi))^2)/2 : \ket\psi\in B(\bbC^d)\} = (1+\|\cN\|_{1\ra
  2}^2)/2$.   By \eq{Sep-ProdSym}, there exists $M'$ with
$h_{\Sep}(M')=(3+\|\cN\|_{1\ra 2}^2)/4$.

\item {\em Estimating $S_\infty^{\min}(\cN)$:}
Since $S_\infty^{\min}(\cN)= -\log \|\cN\|_{1\ra\infty}$, this is
equivalent to  task~\ref{it:1inf-norm}.
\item {\em Estimating $S_2^{\min}(\cN)$:}  Similarly, this is
  equivalent to task~\ref{it:12-norm}.

\item {\em Estimating $\|T\|_{\rm inj}$:}  This relates to $h_{\Sep}$
  in a way that is analogous to the relation between the largest
  singular value of a matrix $A$ and the largest eigenvalue of $A^\dag A$.
\ba \|T\|_{\rm inj}^2
& = \max_{x,y,z\in B(\bbC^d)} \left|\sum_{i,j,k=1}^d T_{i,j,k} x_i y_j
  z_k\right|^2
\nn \\
& = \max_{x,y\in B(\bbC^d)} 
\left\|\sum_{i,j,k=1}^d T_{i,j,k} x_i y_j \ket k \right\|_2^2 
\nn\\
 &= \max_{x,y\in B(\bbC^d)} 
\sum_{i,j,i',j',k=1}^d  T_{i,j,k} T_{i',j',k}^* x_i y_j
x_{i'}^*y_{j'}^* 
\nn\\
& = h_{\Sep}\left(
\sum_{i,j,i',j',k=1}^d  T_{i,j,k} T_{i',j',k}^* \ket i \bra{i'}
\ot \ket j\bra{j'}
\right).
\label{eq:TT-star}\ea
%Let $M$ denote the operator in \eq{TT-star}.
We can think of $T$ as a $d^2\times d$ matrix by grouping indices
$i,j$ together.  By doing so, \eq{TT-star} becomes simply
$h_{\Sep}(TT^\dag)$ (and by our assumption about the operator norm
of the matrix version of  $T$, we have that $TT^\dag \leq I$).  To show the equivalence holds in both
directions, observe that any $M\geq 0$ with rank $\leq d$ can be
written as $TT^\dag$ for some $d^2\times d$ matrix $T$.  This rank
restriction can be removed either by taking $T$ to be a $d\times
d\times d^2$ tensor, or by suitable padding.

\item {\em Estimating $\|T\|_{\ell_2\ra S_\infty}$:} 
Observe that
\[ \|T\|_{\ell_2\ra S_\infty}
 = \max_{x\in B(\bbC^d)} \|T(x)\|_{S_\infty}
 = \max_{x,y,z\in B(\bbC^d)} |\bra y T(x) \ket z|. \]
 This last expression is the maximum of a trilinear form over triples of unit
 vectors, and so is equivalent to computing an injective tensor norm
 (see task~\ref{it:inj}).

\item {\em Estimating $E_{\rm geom}(\ket \psi)$:} Treating $\ket\psi$ as a
  3-index tensor, it is apparent from the definitions that $E_{\rm
    geom}(\ket\psi) = -\log\|\ket\psi\|_{\rm inj}^2$.   The condition
  on $\psi_A$ corresponds to the requirement that $\sqrt{d}$ times the
  resulting tensor should be
  be an isometry when interpreted as a map from $A\ra BC$.   Thus the
  estimation problems are equivalent.  The $\sqrt d$ factor also
  explains why we need to distinguish the cases $E_{\rm geom}\leq
  \log(d/c)$ and $\geq \log(d/s)$.  Interestingly, \thmref{prodtest}
  shows that it is {\em easy} to distinguish whether the geometric
  measurement of entanglement is $\leq \eps$ or $\geq C + \eps$ for a
  sufficiently large constant $C$.

\item {\em Estimating $\nu_\infty(V)$:}  Suppose that $\dim V=m$.
  Define $T$ to be an isometry from $\bbC^m$ to $V$.  Then $\nu_\infty(V)
  = \|T\|_{\ell_2\ra S_\infty}$, and estimating $\nu_\infty(V)$ is
  equivalent to task~\ref{it:2toinf}.  For simplicity, one can assume
  that $m=d$ by padding the appropriate dimensions; this does not
  affect the complexity by more than a polynomial factor. 
% (The astute reader will notice that $\nu_2(V)$ is another equivalent
% problem.   We left it off the list out of a vague sense of restraint.)

\item {\em Mean-field Hamiltonians:} 
In Ref.~\cite{fannes06}, the quantum de Finetti theorem was used
to show that when $n\gg d^2$, then the ground state of $H$ is very
close to a product state.   In the limit, finding the ground-state
energy density of $H$ is equivalent to calculating
the quantity
\[ \max_{\rho \in \cB(d)} \tr K (\rho \ot \rho). \]
This task is therefore equivalent to task 3.

\item {\em Separability testing:} A classic result in convex
  optimisation~\cite{GLS} allows one to show that
  $\WMEM_\eps(\Sep(d,d))$ is roughly equivalent to estimating
  $h_{\Sep}$.  Unfortunately, known versions of this result give up
  $1/\poly(d)$ factors in the approximation guarantees.  
This fact has been used to show the $\NP$-hardness of
$\text{WMEM}_{1/\poly}(\Sep)$ in Refs.~\cite{liu07,gharibian10,beigi10}
and, previously, of $\text{WMEM}_{1/\exp}(\Sep)$ by
Gurvits~\cite{gurvits03} (although the connection to $\QMA_{\log}(2)$
was only observed by \cite{beigi10}).  We conjecture that
$\WMEM_{\eps}(\Sep(d,d))$ should be $\NP_{\log^2}$-hard for some $\eps>0$; i.e. that
$\Sep(d,d)$ cannot be approximated to (sufficiently small) constant
accuracy in time $\poly(d)$.

However, we are able to rule out only algorithms that
have the further restriction of recognizing a nearly convex set that
in turn approximates $\Sep$ to constant accuracy.   The following
result is an immediate consequence of Corollary 4.3.12 of \cite{GLS}
and \corref{3-sat-log}.
\begin{prop}\label{prop:sep-test}
Suppose that there exists a constant $\eps>0$ such that for all $d$,
there exists a convex set $K_d$ with Hausdorff distance $\eps$ to
$\Sep(d,d)$ such that $\WMEM_{1/\poly(d)}(K)$ can
be solved in time $\poly(d)$.
Then $\sat\in\DTIME(\exp(\sqrt{n}\poly\log(n)))$.
\end{prop}
 As a result, one possible alternate title for our paper could have been:

\begin{center}
\fbox{{\em Detecting pure entanglement is easy, so detecting mixed entanglement is hard.}}
\end{center}

In fact, reductions between WMEM and approximating $h_{\Sep}$ go in both directions.  We
have to be careful not to assume (as does \cite{ioannou07}) that
approximation algorithms for $h_{\Sep}$
output an approximately optimal density matrix.  Indeed, some
approximations (e.g.\ \cite{brandao10a}) only output a scalar value
approximating $h_{\Sep}$.  However, we can prove the following reduction.

\begin{prop}\label{eq:WMEM-from-WOPT}
Let $f(M)$ be a convex function such that $f(0)=0$ and $|f(M)-h_{\Sep(d,d)}(M)|\leq
\eps \|M\|_\infty$.  Given oracle access to $f$, we can solve
$\WMEM_{2\eps}(\Sep(d,d))$ in time $\poly(d)$.
\end{prop}

\begin{proof}
Suppose we are given a density matrix $\rho$ for which we would like
to solve $\WMEM_\eps(\Sep(d,d))$. The algorithm computes  
$$Z:=\max\{ \tr M \rho - f(M) : -I \leq M \leq I\}.$$  This can be done in polynomial
time~\cite{GLS}.    If $Z\leq \eps$, then we declare that $\rho\in
\Sep(d,d)$, and if $Z>\eps$ then we declare that $\rho\notin
B_1(\Sep(d,d),\eps)$.

To analyze the correctness of the algorithm, we prove rather that it
is {\em not wrong}.  In other words, we need to give the correct
answer in the cases: (1) when $\rho\not\in B_1(\Sep(d,d),2\eps)$, and 
(2) when $\rho\in \Sep(d,d)$.  In case (1), then $\rho$ has
trace distance $>2\eps$ from every point in $\Sep(d,d)$ and so there
exists a $M$ with $\|M\|_\infty\leq 1$ for which $\tr M\rho >
h_{\Sep}(M) + 2\eps$.  This implies that $\tr M\rho > f(M)+\eps$, and
that $Z>\eps$.

On the other hand, in case (2), we have $\rho\in \Sep(d,d)$, which
implies $\tr M\rho \leq h_{Sep(d,d)}(M) \leq f(M)+\eps$ for all $M$,
and thus $Z\leq \eps$.
\end{proof}

\item {\em Weak membership for entanglement witnesses:}  For a
  Hermitian matrix $M$, we have $h_{\Sep}(M)\leq \eps$ if and only if
  $M \in B_\infty({\rm EW},\eps)$.   Here $B_\infty(S,\eps)$ refers to the points
  within $\eps$ of a set $S$ in the Schatten-$\infty$ norm.
 This shows
  that if we can approximate $h_{\Sep}$ then we can solve the weak
  membership problem for EW.  Conversely, if we are given an algorithm
 for $\WMEM_\eps^\infty({\rm EW})$, then on input $M$ we can use binary
 search to find approximately the smallest $\gamma$ such that
 $M-\gamma I\in {\rm EW}$.  This $\gamma$ will be within $\epsilon$ of
 $h_{\Sep}(M)$.

\item {\em $k$-partite tensor norm problems:}  By adding more systems,
  we will not make any of the problems any easier.   To reduce from an
  injective tensor norm on $k$-tensors to the injective tensor norm on
  3-tensors, we can use \lemref{qmak-sim}.   The other reductions claimed
  in this point are similar.

When performing these mappings, there is no direct penalty that
depends on $k$.   However, the dimensions of the spaces involved will
scale exponentially with $k$.  For example, estimating the support
function of $\Sep(d_1,d_2,\ldots,d_k)$ is harder than estimating
$h_{\Sep(d_1,d_2)}$, and by \lemref{qmak-sim} can be reduced to estimating
$h_{\Sep(d,d)}$, where $d:=\poly(d_1d_2\cdots d_k)$.

\item {\em $\ell_2 \rightarrow \ell_4$ norm:} If $A= \sum_{i=1}^m \ket i \bra{\alpha_i}$ with each $\ket{\alpha_i}\in\bbC^n$, then
\be \|A\|_{2\ra 4}^4 = \max_{\ket\psi\in S(\bbC^n)} |\braket{\alpha_i}{\psi}|^2
 = h_{\ProdSym}(\sum_i \alpha_i^{\ot 2})
  %= h_{\Sep}(\sum_i \alpha_i^{\ot 2})
  \ee
To show a reduction in the other direction, we need to convert any measurement $M$ into an $M'$ with similar $h_{\Sep}$ such that $M'$ can be written as $\sum_i \alpha_i \ot \alpha_i$.
 \thmref{hSep-equiv} instead allows us to write $M$ in the form $\sum_i \alpha_i \ot \beta_i$, but the desired $\sum_i \alpha_i \ot \alpha_i$ form can be achieved by taking
  $M' =  (\sqrt{M}^{A_1B_1} \ot \sqrt{M}^{A_2B_2}) (P_{\rm sym}^{A_1B_1} \ot P_{\rm sym}^{A_2B_2})(\sqrt{M}^{A_1B_1} \ot \sqrt{M}^{A_2B_2})$. This construction is analyzed in \cite{BBHKSZ12}.
  
  As a result, it is $\NP$-hard to approximate the $2\ra 4$ norm to $1/\poly(d)$ accuracy, and it is $\NP_{\log^2}$-hard to approximate it to within a constant multiplicative error.

\item {\em $\sat$ on $\log^2$ variables:} \corref{3-sat} explains how
  $\sat$ instances with $O(\log^2(n))$ variables can be reduced to
  determining whether $h_{\Sep(n,n)}(M)=1$ or is $\leq 1/2$, for some
  efficiently-computable matrix $M$.  It is an extremely interesting
  open question to determine whether the reverse reduction is also
  possible. 

\item {\em Planted clique problem:}  In \cite{BV09}, the problem of
  finding a clique of size $\Omega(n^{1/r}r^5\log^3(n)\alpha^2)$
  planted in a $G_{n,1/2}$ graph 
was reduced to determining whether $\|T\|_{\rm inj}$ is $\geq \alpha^r$ or
$\leq 1$.  Here $\alpha\leq 1$ and $T$ is a tensor in
$(\bbC^n)^{\ot r}$ with all $\pm 1$ entries.  Thus, we can trivially
bound $\|TT^\dag\|_\infty \leq n^r$ (although we suspect that the norm
should typically be $O(n^{r-1})$).

For concreteness, let us focus on the case of $r=3$.  In this case,
\cite{BV09} reduce the problem of finding a clique of size
$n^{1/3+o(1)}$ to $\QMA_{\log}(2)_{1/n^{1.5},2/n^{1.5}}$ (or
$\QMA_{\log}(2)_{1/n,2/n}$, if one assumes that $\bbE \|TT^\dag\|_\infty
\leq O(n^{r-1})$).   Since random graphs
typically have no clique of size larger than $2\log n+1$, the planted
clique problem can always be reduced to a Circuit-SAT instance of size
$\poly(n)$ with $O(\log^2(n))$ input variables.
Since 
$$\QMA_{\log}(2)_{1/n^{1.5},2/n^{3/2}} \supseteq \QMA_{\log}(2)_{1/2,1}\ni 
\sat_{\log^2},$$
this implies that the reduction of \cite{BV09} achieves a reduction
that is comparable to the previously known reduction to Circuit-SAT.
It is an interesting open question to determine whether Circuit-SAT
instances with size $\poly(n)$ and $O(\log^2(n))$ input variables can
be placed in $\QMA_{\log}(2)_{1/2,1}$.  If this were possible, then it
would imply that the reduction of \cite{BV09} would be strictly
subsumed by the previous reduction of planted clique to Circuit-SAT.

\item {\em Minimum output R\`enyi entropies:} For any $\alpha\geq 0$, we have $S_\alpha^{\min}(\cN) \geq S_\infty^{\min}(\cN)$ but also $S_\alpha^{\min}(\cN)=0$ iff $S_\infty^{\min}(\cN)=0$. Thus, for any $c>0$, distinguishing between $S_\alpha^{\min}=0$ and $S_\alpha^{\min} \ge c$ is at least as hard as distinguishing between $S_\infty^{\min}=0$ and $S_\infty^{\min}\ge c$.

\item {\em Regularised minimum output R\`enyi entropies:} Our hardness result
for $S_\alpha^{\min}$ immediately gives us the equivalent hardness result for $S_\alpha^{R, \min}$.
The reason is that our proof of amplification for $\QMA(2)$ protocols (see \lemref{qma-sep-ampl}) essentially works by
constructing a channel $\cN$ for which $S_\infty^{R, \min}(\cN) = S_\infty^{\min}(\cN)$ by design.

\end{enumerate}

\subsubsection{Additional remarks}

\begin{itemize} 
\item {\em Additivity violations:} As a result of the connection between $\QMA(2)$ and
estimating $S_\infty^{\min}$, the question of whether $\QMA(2)$
protocols can be amplified to exponentially small error is directly
related to the question of additivity of the minimum output
min-entropy (equivalently, multiplicativity of the maximum output
infinity norm).  Indeed, additivity violations for $S_\infty^{\min}$
(e.g.\ \cite{HW02,hayden08,GHP09}) translate directly into $\QMA(2)$
protocols for which perfect parallel repetition fails\footnote{Note
  that, taking the standard definition of $\QMA(2)$, this is strictly
  speaking only true if the corresponding $\QMA(2)$ protocol can be
  implemented in polynomial time.}.  Conversely, \cite{Mat05} observed
that $\QMA(2)$ protocols obey perfect parallel repetition when the
corresponding channel $\cN$ is known to have additive
$S_\infty^{\min}$, for example when $\cN$ is entanglement breaking.
Indeed our \lemref{qma-sep-ampl} is a restatement of this point.

\item {\em Minimum output entropy:}
Beigi and Shor previously showed that it is $\NP$-hard to compute the
minimum output entropy up to $1/\poly(d)$ accuracy \cite{beigi07}. Our
result improves their accuracy requirement, but under a more
restrictive complexity assumption.  
For general channels, we automatically have $S_\alpha^{R, \min}(\cN) \leq
S_\alpha^{\min}(\cN)$; however, the famous failures of the additivity
conjecture imply that sometimes this inequality can be strict, with
examples known for $\alpha\geq 1$ \cite{Hastings,hayden08} and for
$\alpha$ near 0 \cite{cubitt08a}.  Still, these examples only
demonstrate that $S_\alpha^{R, \min}$ can deviate very slightly from
$S_\alpha^{\min}$.  On the other hand, various lower bounds for
$S_\alpha^{R, \min}$ are known \cite{vidal01,yang05,devetak06,YHW08},
and it may be that one of these bounds could be related to
$S_\alpha^{\min}$, thereby proving that $S_\alpha^{R, \min}$ cannot be
far from $S_\alpha^{\min}$.  Our results do not rule out the
possibility that $S_\alpha^{\min}$ may be fruitfully related to
$S_\alpha^{R, \min}$.  However, they do imply that these lower bounds
on $S_\alpha^{R, \min}$ (and thereby on $S_\alpha^{\min}$) are
unlikely to be efficiently computable, or if they are, they are likely
to be extremely loose bounds in general. 

\item {\em Mean-field approximation:} Previous work on the hardness of
  approximating ground-state energy of quantum systems generally had
  $d$ constant and only ruled out the possibility of $1/\poly(n)$
  approximation error.  By addressing the case when approximation
  error is a constant fraction of the overall energy of the system,
  our result achieves one of the goals of the conjectured quantum PCP
  theorem \cite{AALV09}.  However, we require $d$ to grow
  asymptotically, and we achieve a hardness result much weaker than
  QMA-hardness.  Indeed, due to the {\em classical} PCP theorem
  combined with the Exponential Time Hypothesis, finding the ground
  state of a system of $d^2\log(d)$ bits (without any symmetry
  constraint) is likely to require time $\exp(d^2\log(d))$, while our
  results merely imply an $\exp(\Omega(\log^2(d)))$ lower bound.
  Still, our result provides a superpolynomial bound on an important
  class of Hamiltonians that had been previously considered to be
  computationally easy to work with.

\end{itemize}

%-------------------------------------------------------------------------------

\section{Optimality of the product test}
\label{sec:optimal}

Our product test has perfect completeness in the sense that if $\ket{\psi}$ is
exactly a product state then it will always pass the product
test.  Soundness could be in principle described by a functional
relation between maximum acceptance probability and distance to the
nearest product state.  However, for our purposes, we can say that our test
has constant soundness in that if $\ket{\psi}$ has overlap at most $1-\eps$
with any product state then it will pass the product
test with probability at most $1-\Theta(\eps)$.

In fact, if we consider only product-state tests with perfect
completeness, then we can show that our test has optimal soundness:
that is, it rejects as often as possible given the constraint of
always accepting product states.  More generally, suppose that a product-state
test $T$ is given $\ket{\psi}^{\ot k}$ as input.  Since the outcome of the
test is binary, we can say that $T$ is an operator on the $nk$-qudit
Hilbert space with $0\leq T \leq I$ and that the test accepts with
probability $\tr T \psi^{\ot k}$.  

Let $S$ be the set of product states in $\bbC^{d_1}\ot \cdots \ot
\bbC^{d_n}$, and define $S^k$ to be the span of $\{\ket{\phi}^{\ot k}:\ket\phi \in S\}$.
For a single system $\bbC^d$, the span of $\{\ket{\phi}^{\ot k} :
\ket\phi\in\bbC^d\}$ is denoted $\Sym^k \bbC^d$.  This is the symmetric
subspace of $(\bbC^d)^{\ot k}$, meaning that it can be equivalently
defined to be the set of vectors in $(\bbC^d)^{\ot k}$ that is
invariant under permutation by the symmetric group $S_k$.  This fact
allows the projector onto $\Sym^k\bbC^d$, which we denote
$\Pi^{\text{sym}}_{d,k}$, to be implemented efficiently~\cite{barenco97}.
Also, it implies that $S^k = \Sym^k \bbC^{d_1} \ot \cdots \ot
\Sym^k \bbC^{d_n}$ and that the projector onto $S^k$, denoted $\Pi_{S^k}$, is
$\Pi^{\text{sym}}_{d_1,k} \ot \cdots \ot \Pi^{\text{sym}}_{d_n,k}$.
% See \cite{Har-sym} for a review of the properties of $\Sym^k\bbC^d$.

Now we return to our discussion of product-state tests.  If $\tr T
\phi^{\ot k}=1$ for all $\phi\in S$, then $T \geq \Pi_{S^k}$.  Thus,
$T$ will always accept at least as often as $\Pi_{S^k}$ will on any
input, or equivalently, taking $T=\Pi_{S^k}$ yields the test which
rejects as often as possible given the constraint of accepting every
state in $S^k$.

To understand $\Pi_{S^k}$, note that the projector onto $\Sym^k\bbC^d$
is given by 
$\frac{1}{k!}\sum_{\pi\in\cS_k}P_d(\pi)$, where
\be P_d(\pi)=\sum_{i_1,\ldots,i_k\in [d]} \ket{i_1,\ldots,i_k}
\bra{i_{\pi(1)},\ldots,i_{\pi(k)}}.
\label{eq:Pmatrix}\ee
For $k=1$, $\Sym^1 \bbC^d$ simply equals
$\bbC^d$, and $\Pi_{S^1}$ is the identity operator on $(\bbC^d)^{\ot
  n}$.  Thus, no non-trivial product-state test is possible when given one
copy of $\ket\psi$.

When $k=2$, $\Sym^2 \bbC^d$ is the $+1$ eigenspace of $(I+\cF)/2$,
which is the space that passes the swap test.  Thus, the product
test (in \protoref{prodtest}) performs the projection onto
$S^2$ and therefore rejects non-product states as often as possible
for a test on $\ket{\psi}^{\ot 2}$ that always accepts when $\ket\psi$
is a product state.
These arguments also imply that given  $\ket\psi^{\ot k}$,
projecting onto $S^k$ yields an optimal $k$-copy product-state test of
$\ket\psi$.  The strength of these tests is strictly increasing with
$k$, but we leave the problem of analysing them carefully to future
work.

Finally, this interpretation of the product test allows us to
consider generalisations to testing membership in other sets $S$.  The
general prescription for a test that is given $k$ copies of a state is
simply to project onto the span of $\{\ket\psi^{\ot k}: \ket\psi\in S\}$.
We will not explore these possibilities further in this
paper, but see \cite{wang11} for a subsequent paper that considers
variations on this theme for the related problem of
testing properties of unitary operators.

\section{Stability of the depolarising channel}
\label{sec:depolar}

As discussed in \secref{strategy}, the correctness proof of the
product test in fact applies to a larger class of processes acting on
two copies of $n$-partite states.  In general, choosing $S$ according
to a binomial distribution on $[n]$ and taking the expectation of
$\tr\rho_S^2$ is equivalent to evaluating the output purity of $n$
uses of the depolarising channel $\cD_\delta$ (as defined in
(\ref{eqn:depolarising})).  A special case of this corresponds to the
probability of the product test passing.

\newcounter{depolarise}\setcounter{depolarise}{\value{thm}}
\begin{lem}
\label{lem:depolarise}
We have
\[ \tr (\mathcal{D}_{\delta}^{\ot n}\, \rho)^2 = \left(\frac{1-\delta^2}{d}\right)^n \sum_{S \subseteq [n]} \left(\frac{d \delta^2}{1-\delta^2}\right)^{|S|} \tr(\rho_S^2), \]
and in particular
\[ \tr (\mathcal{D}_{1/\sqrt{d+1}}^{\ot n}\, \rho)^2 = \frac{1}{(d+1)^n}\sum_{S \subseteq [n]} \tr(\rho_S^2), \]
and for pure product states,
\[
\Pprod(\delta) := \tr (\mathcal{D}_{\delta}^{\ot n}\, (\proj{\psi_1} \ot \cdots \ot \proj{\psi_n}))^2
= \left(\frac{d-1}{d} \delta^2 + \frac{1}{d} \right)^n.
\]
\end{lem}

%If $d_1 = d_2 = \dots = d_n = d$, for some $d$, then
%\[ \Ptest(\rho) = \left( \frac{d+1}{2} \right)^n \tr (\mathcal{D}_{1/\sqrt{d+1}}^{\ot n}\, \rho)^2. \]

We will see that it is possible to prove a more general version of
\thmref{prodtest}. 

\newcounter{stability}\setcounter{stability}{\value{thm}}
\begin{thm}
\label{thm:stability}
Given $\ket{\psi}\in (\bbC^d)^{\ot n}$, let
\[ 1 - \eps = \max\{ |\braket{\psi}{\phi_1,\ldots,\phi_n}|^2 : \ket{\phi_1},\ldots,\ket{\phi_n}\in\bbC^d\}.\]
Then (recalling the definitions of $\cD_\delta$ and $\Pprod(\delta)$ from
equations (\ref{eqn:depolarising}) and (\ref{eq:OPP})),
\[\tr (\mathcal{D}_{\delta}^{\ot n}\, \proj{\psi})^2 \le \Pprod(\delta)\Bigg(1 - 4\eps(1-\eps)\frac{d \delta^2(1-\delta^2)}{(1\!+\!(d-1)\delta^2)^2}+ 4 \eps^{3/2} \left(\frac{(1-\delta^2)^2+ d^2 \delta^4}{(1\!+\!(d-1)\delta^2)^2}\right)^2 \Bigg).
\]
In particular,
\[ \tr (\mathcal{D}_{1/\sqrt{d+1}}^{\ot n} \proj{\psi})^2 \le \Pprod(1/\sqrt{d+1})\left(1\!-\!\eps\!+\!\eps^2\!+\!\eps^{3/2} \right).\]
\end{thm}

The idea of the proof is more or less the same as the outline sketched
in \secref{strategy} and the details can be found in
Appendix~\ref{sec:depolarising}.

%-------------------------------------------------------------------------------
%\section{Further applications}
\section{Testing for product unitaries}
\label{sec:unitaries}

As well as being useful for testing quantum states, the product test has applications to testing properties of unitary operators.   
The results we obtain will be in terms of the normalised Hilbert-Schmidt inner product, which is defined as $\hsip{M}{N} := \frac{1}{d} \tr M^\dag N$ for $M,N \in M(d)$, where $M(d)$ denotes the set of $d \times d$
matrices. Note that, with this normalisation, $|\hsip{U}{V}| \le 1$
for unitary operators $U$, $V$. 

We consider the problem of testing whether a unitary operator is a tensor product. That is, we are given access to a unitary $U$ on the space of $n$ qudits (for simplicity, restricting to the case where each of the qudits has the same dimension $d$), and we would like to decide whether $U = U_1 \otimes \cdots \otimes U_n$. This is one possible generalisation of the classical problem of testing linearity of functions $f:\{0,1\}^n \rightarrow \{0,1\}$ \cite{blum93}. To see this, observe that $f$ is linear (i.e.\ $f(x \oplus y) = f(x) \oplus f(y)$ for all $x$ and $y$) if and only if the function $g:\{0,1\}^n \rightarrow \{\pm 1\}$ defined by $g(x) = (-1)^{f(x)}$ is a product of individual functions $g_i(x) = (-1)^{a_i x_i}$, for $a_i \in \{0,1\}$. Thus the diagonal unitary operator $U$ on $n$ qubits defined by $U_{xx} = g(x)$ is a tensor product if and only if $f$ is linear.

In \protoref{produnitarytest} we give a test that solves this
problem using the product test. The test always accepts product
unitaries, and rejects unitaries that are far from product, as
measured by the normalised Hilbert-Schmidt inner product. 
Several papers have proposed property tests with similar performance for other sets of unitary matrices: e.g.\ Pauli matrices~\cite{qboolean}, Clifford gates~\cite{cliff-test,wang11} and many other sets~\cite{wang11}.

The following correspondence (also
known as the Choi-Jamio\l{}kowski isomorphism) underlies our ability
to apply the product test to unitaries. 
Let $\ket{\Phi}$ be a maximally entangled state of two $d$-dimensional
qudits, written as $\frac{1}{\sqrt{d}}\sum_{i=1}^d \ket{i,i}$ in terms
of some basis $\mathcal{B} = (\ket{1},\dots,\ket{d})$. For any matrix
$M \in M(d^n)$, define $\ket{v(M)} := (M \otimes I)
\ket{\Phi}^{\otimes n}$. Then $\bra{j}\braket{k}{v(M)} =
\frac{\braket{j}{M|k}}{\sqrt{d^n}}$. In particular, for any matrices
$M,N \in M(d^n)$, $\hsip{M}{N} = \braket{v(M)}{v(N)} = \tr M^\dag N / d^n$.

%\begin{proof}
%This is just the well-known Choi-Jamio\l{}kowski isomorphism. 
%Written out explicitly, we have
%\beas
%\bra{j}\braket{k}{(M \otimes I)|\Phi}^{\otimes n} &=& \bra{j}\bra{k}(M \otimes I) \left(\frac{1}{\sqrt{d}} \sum_{i=1}^d \ket{i,i}_{AB} \right)^{\otimes n}\\
%&=& \frac{1}{\sqrt{d^n}} \sum_{i_1,\dots,i_n=1}^d \braket{j}{M|i_1,\dots,i_n} \braket{k}{i_1,\dots,i_n} = \frac{\braket{j}{M|k}}{\sqrt{d^n}}.
%\eeas
%
%The second claim in the lemma follows immediately from the first.
%\end{proof}

\boxproto{Product unitary test}{
\label{proto:produnitarytest}
The product unitary test proceeds as follows.
\begin{enumerate}
\item Prepare two copies of the state $\ket{\Phi}^{\otimes n}$, then in both cases apply $U$ to the $n$ first halves of each pair of qudits to create two copies of the state $\ket{v(U)} \in (\bbC^{d^2})^{\ot n}$.
\item Return the result of applying the product test to the two copies of $\ket{v(U)}$, with respect to the partition into $n$ $d^2$-dimensional subsystems.
\end{enumerate}
}
Let the probability that this test passes when applied to some unitary
$U$ be $\Ptest(U)$. Then we have the following theorem, which proves a conjecture from \cite{qboolean}.

\def\thmunitarytest{
Given $U\in U(d^n)$, let
\[ 1 - \eps = \max\{ \left|\hsip{U}{V_1 \otimes \dots \otimes V_n}\right|^2 : V_1, \dots, V_n \in U(d)\}.\]
Then, if $\eps = 0$, $\Ptest(U)=1$. If $\eps \lesssim 0.106$, then  $\Ptest(U) \le 1 - \frac{1}{4} \eps + \frac{1}{16}\eps^2 + \frac{1}{8}\eps^{3/2}$. If $0.106 \lesssim \eps \le 1$, $\Ptest(U) \le 501/512$. More concisely, $\Ptest(U) = 1 - \Theta(\epsilon)$.
}
\begin{thm}\label{thm:unitary-test}
\thmunitarytest
\end{thm}

The proof is given in Appendix \ref{sec:unitary-proof}. It is not quite immediate from the previous results; the key problem is that the closest product state to $\ket{v(U)}$ may not correspond to the closest unitary operator to $U$.

Our test is sensitive to the Hilbert-Schmidt distance of a unitary
from the set of product unitaries. One might hope to design a similar test
that instead uses a notion of distance based on the operator norm.
However, this is not possible.
For example, if we could detect a constant difference in the operator norm
between an arbitrary unitary $U$ and the set of product unitaries then
we could find a single marked item in a set of size $d^n$.  By the
optimality of Grover's algorithm, this requires $\Omega(d^{n/2})$
queries to $U$.   More generally, any test that uses only a constant number
of black-box queries to $U$ can only detect an $\Omega(1)$ difference in an
$\Omega(1)$ fraction of the $d^n$ dimensions that $U$ acts upon.

%-------------------------------------------------------------------------------
\section{Open problems}
\label{sec:conc}

We conclude with a discussion of open problems related to our work.
\begin{enumerate}
\item Our main result can be seen as a ``stability'' theorem for the output
purity of the depolarising channel (cf. \secref{depolar})..  It is an
interesting problem to 
determine whether a similar result holds for all output R\'enyi
entropies for the depolarising channel, or even for all channels where
additivity holds. 
\item Can Theorem \ref{thm:prodtest} be tightened further, perhaps by improving the
constant in the $\epsilon^{3/2}$ term?  It would also be interesting
to improve the constants in Theorem \ref{thm:prodtest} in the regime
of large $\epsilon$, as at present they are extremely pessimistic.
The regime of large $\eps$ is generally somewhat mysterious: for
example, we do not know the minimum value of $\Ptest$, or the largest
distance from any product state that can be achieved by a state of $n$
qudits. This is equivalent to determining the maximal value of the geometric
measure of entanglement~\cite{wei03} which can be achieved by a pure state of $n$ qudits;
see the PhD thesis~\cite{aulbach11} for a review of recent work concerning bounds on this quantity.
\item Suppose the goal is to test whether a given state is of the form $\ket{\varphi}^{\ot n}$ for some unknown $\ket\varphi$.  Can we substantially improve on the performance of the product test, say with a test whose acceptance probability decreases exponentially in the number of positions not equal to $\ket\varphi$?  Ideally we would achieve performance comparable to the exponential de Finetti theorem~\cite{renner05}, but without any dependence on dimension.  The natural test for this problem is to project onto the symmetric subspace of all $2n$ positions.
\item The relationship between $\QMA$ and $\QMA(2)$ remains unresolved.  Our \thmref{qma2k} proves that $\QMA^{\SEP}(2)=\QMA(2)$, while the result of \cite{brandao11} implies that $\QMA^{{\rm LOCC}}(2)=\QMA$.  Can this gap be closed?  One possible way to do this would be to improve our results to show that $\QMA^{{\rm LOCC}}(2)=\QMA(2)$; but see also Appendix~\ref{sec:locc} for a proof that an efficient LOCC product test does not exist.  Alternatively, one might improve the simulation of \cite{brandao11} to apply to separable measurements instead of only LOCC measurements, but the obvious approaches to modifying their proof do not appear to work.  Finally, if $\QMA(2)$ is not shown to be in $\QMA$, one might hope for any upper bound on its complexity that is better than $\NEXP$.
\item Is there an oracle separation between $\QMA$ and $\QMA(2)$?   The equalities in the previous point relativise, so this is equivalent to showing a separation between $\QMA^{\SEP}(2)$ and $\QMA^{\rm LOCC}(2)$.
\end{enumerate}

%-------------------------------------------------------------------------------

\appendix

\section{The depolarising channel}
\label{sec:depolarising}

Let $\mathcal{D}_\delta$ be the qudit depolarising channel as defined
in equation (\ref{eqn:depolarising}). We will be interested in applying the
$n$-fold product $\mathcal{D}_\delta^{\ot n}$ to states of $n$ qudits,
and in particular in the purity of the resulting states. This has the
following characterisation. 

\newcounter{skip}\setcounter{skip}{\value{thm}}
\setcounter{thm}{\value{depolarise}}

\begin{lem}
We have
\[ \tr (\mathcal{D}_{\delta}^{\ot n}\, \rho)^2 = \left(\frac{1-\delta^2}{d}\right)^n \sum_{S \subseteq [n]} \left(\frac{d \delta^2}{1-\delta^2}\right)^{|S|} \tr(\rho_S^2), \]
and in particular
\[ \tr (\mathcal{D}_{1/\sqrt{d+1}}^{\ot n}\, \rho)^2 = \frac{1}{(d+1)^n}\sum_{S \subseteq [n]} \tr(\rho_S^2), \]
and for pure product states,
\[ \Pprod(\delta) := \tr (\mathcal{D}_{\delta}^{\ot n}\, (\proj{\psi_1} \ot \cdots \ot \proj{\psi_n}))^2 = \left(\frac{d-1}{d} \delta^2 + \frac{1}{d} \right)^n.\]
\end{lem}

\begin{proof}
Consider some Hermitian operator basis for $\mathcal{B}(\bbC^d)$ which contains the identity and is orthonormal with respect to the normalised Hilbert-Schmidt inner product $\langle A,B \rangle = \frac{1}{d}\tr A^\dag B$, and extend this basis to $\mathcal{B}((\bbC^d)^{\otimes n})$ by tensoring. Expand $\rho$ in terms of the resulting basis as
\[ \rho = \sum_{\mathbf{t} \in \{0,\dots,d^2-1\}^n} \hat{\rho}_{\mathbf{t}} \chi_{\mathbf{t}}.\]
where $\hat{\rho}_{\mathbf{t}} \in \bbR$, $\chi_{\mathbf{t}}$ represents an element of the tensor product basis corresponding to the string $\mathbf{t} \in \{0,\dots,d^2-1\}^n$, and the identity is indexed by 0 at each position. Then we have
\[ \tr(\rho_S^2) = d^{2n-|S|}\left(\sum_{\mathbf{t},\,\mathbf{t}_i=0,\,\forall i \in \bar{S}} \hat{\rho}_{\mathbf{t}}^2\right), \]
and hence, for any $\delta$,
\beas
\sum_{S \subseteq [n]} \delta^{|S|} \tr(\rho_S^2) &=& d^{2n} \sum_{S \subseteq [n]} (\delta/d)^{|S|} \left( \sum_{\mathbf{t},\,\mathbf{t}_i=0,\,\forall i \in \bar{S}} \hat{\rho}_{\mathbf{t}}^2 \right) = d^{2n} \sum_{\mathbf{t}} \hat{\rho}_{\mathbf{t}}^2 \left( \sum_{\substack{S \subseteq [n],\\\mathbf{t}_i=0,\,\forall i \in \bar{S}}} (\delta/d)^{|S|} \right)\\
&=& d^{2n} \sum_{\mathbf{t}} \hat{\rho}_{\mathbf{t}}^2 \left( \sum_{x=0}^{n-|\mathbf{t}|} \binom{n-|\mathbf{t}|}{x}\,(\delta/d)^{x+|\mathbf{t}|} \right)\\
&=& d^{2n} \sum_{\mathbf{t}} \hat{\rho}_{\mathbf{t}}^2\,(\delta/d)^{|\mathbf{t}|}(1+\delta/d)^{n-|\mathbf{t}|} \\
&=& (d(d+\delta))^n \sum_{\mathbf{t}} \hat{\rho}_{\mathbf{t}}^2\,(\delta/(\delta+d))^{|\mathbf{t}|} \\
&=& (d+\delta)^n \tr (\mathcal{D}_{\sqrt{\delta/(\delta+d)}}^{\ot n} \rho)^2.
\eeas
Rearranging completes the proof; the two special cases in the statement of the lemma can be verified directly.
\end{proof}

Using the above lemma, we can see that maximal output purity is
obtained only for product states, since only product states saturate
the inequality $\tr \rho_S^2 \leq 1$ for all $S\subseteq [n]$. We will
now prove our main result, which is a ``stability'' theorem for the
depolarising channel: if a state achieves close to maximal output
purity, it must be close to a product state. 

\setcounter{thm}{\value{stability}}

\begin{thm}
Given $\ket{\psi}\in (\bbC^d)^{\ot n}$, let
\be 1 - \eps = \max\{ |\braket{\psi}{\phi_1,\ldots,\phi_n}|^2 : \ket{\phi_1},\ldots,\ket{\phi_n}\in\bbC^d\}.\label{eq:closest-prod2}\ee
Then 
\[ \tr (\mathcal{D}_{\delta}^{\ot n}\, \proj{\psi})^2 \le \Pprod(\delta)\left(1 - 4\eps(1-\eps)\frac{d \delta^2(1-\delta^2)}{(1+(d-1)\delta^2)^2} + 4 \eps^{3/2} \left(\frac{(1-\delta^2)^2+ d^2 \delta^4}{(1+(d-1)\delta^2)^2}\right)^2 \right).\]
In particular,
\[ \tr (\mathcal{D}_{1/\sqrt{d+1}}^{\ot n}\, \proj{\psi})^2 \le \Pprod(1/\sqrt{d+1})\left(1 - \eps + \eps^2 + \eps^{3/2} \right).\]
\end{thm}

\begin{proof}
Without loss of generality assume that one of the states achieving the maximum in \eq{closest-prod2} is $\ket{0}^{\ot n}$, which we will abbreviate simply as $\ket{0^n}$, or $\ket 0$ when there is no ambiguity. We thus have
\[ \ket{\psi} = \sqrt{1-\eps} \ket{0} + \sqrt{\eps} \ket{\phi} \]
for some state $\ket{\phi}$ such that $\braket{0}{\phi}=0$, and $\ket{\phi} = \sum_{x \neq 0} \alpha_x \ket{x}$ for some $\{\alpha_x\}$. We write down explicitly
\[ \psi := \proj{\psi} = (1-\eps) \proj{0} + \sqrt{\eps(1-\eps)}(\ket{0}\bra{\phi} + \ket{\phi}\bra{0}) + \eps \proj{\phi}. \]
By Lemma \ref{lem:depolarise},
\[ \tr (\mathcal{D}_{\delta}^{\ot n}\, \psi)^2 = \left(\frac{1-\delta^2}{d}\right)^n \sum_{S \subseteq [n]} \gamma^{|S|} \tr\psi_S^2, \]
where we set $\gamma = d \delta^2/(1-\delta^2)$ for brevity. Now
\[ \sum_{S \subseteq [n]} \gamma^{|S|} \tr\psi_S^2 = \sum_{S \subseteq [n]} \gamma^{|S|} \left( \tr ((1-\eps) \ket{0}\bra{0}_S + \sqrt{\eps(1-\eps)}(\ket{0}\bra{\phi}_S + \ket{\phi}\bra{0}_S) + \eps \ket{\phi}\bra{\phi}_S)^2 \right), \]
and for any subset $S$,
\beas
\tr\psi_S^2 &=& (1-\eps)^2 \tr \proj{0}_S^2 + \eps(1-\eps) \tr(\ket{0}\bra{\phi}+\ket{\phi}\bra{0})_S^2 + \eps^2 \tr \proj{\phi}_S^2\\ &+& 2 \sqrt{\eps} (1-\eps)^{3/2} \tr \proj{0}_S(\ket{0}\bra{\phi}+\ket{\phi}\bra{0})_S + 2 \eps(1-\eps) \tr \proj{0}_S \ket{\phi}\bra{\phi}_S\\ &+& 2 \eps^{3/2}\sqrt{1-\eps} \tr \proj{\phi}_S(\ket{0}\bra{\phi}+\ket{\phi}\bra{0})_S.
\eeas
We now bound the sum over $S$ (weighted by $\gamma^{|S|}$) of each of these terms, in order. Note that we repeatedly use the notation $[E]$ for a term which evaluates to 1 if the expression $E$ is true, and 0 if $E$ is false.
\begin{enumerate}
\item As $\ket{0}$ is product, clearly
\[ \sum_{S \subseteq [n]} \gamma^{|S|} \tr \proj{0}_S^2 = \sum_{S \subseteq [n]} \gamma^{|S|} = (1 + \gamma)^n. \]
\item We have
\[ \tr(\ket{0}\bra{\phi}+\ket{\phi}\bra{0})_S^2 = \tr \ket{0}\bra{\phi}_S^2 + \tr \ket{\phi}\bra{0}_S^2 + 2 \tr \ket{0}\bra{\phi}_S \ket{\phi}\bra{0}_S. \]
It is easy to see that the first two terms must be 0 for all $S$ (as only the off-diagonal entries of the first row of the matrix $\ket{0}\bra{\phi}$ can be non-zero). For the third, we explicitly calculate
\[ \ket{0}\bra{\phi}_S \ket{\phi}\bra{0}_S = \sum_{x \neq 0} |\alpha_x|^2 [x_i = 0, \forall i \in \bar{S}] \proj{0}^{\ot k}, \]
and hence
\beas
\sum_{S \subseteq [n]} \gamma^{|S|} \tr \ket{0}\bra{\phi}_S \ket{\phi}\bra{0}_S &=& \sum_{x \neq 0} |\alpha_x|^2 \sum_{S \subseteq [n]} \gamma^{|S|} [x_i = 0, \forall i \in \bar{S}]\\
&=& \sum_{x \neq 0} |\alpha_x|^2 \sum_{k=|x|}^n \gamma^k \binom{n-|x|}{n-k}\\
&=& (1+\gamma)^n \sum_{x \neq 0} |\alpha_x|^2 \left(\frac{\gamma}{1+\gamma}\right)^{|x|}.
\eeas

\item It clearly holds that $\tr \proj{\phi}_S^2 \le 1$, so as in part (1),
\[ \sum_{S \subseteq [n]} \gamma^{|S|} \tr \proj{\phi}_S^2 \le (1 + \gamma)^n, \]
and this will be tight if and only if $\ket{\phi}$ is product itself.

\item Using the same argument as in part (2), $\tr \proj{0}_S \ket{0}\bra{\phi}_S = \tr \proj{0}_S \ket{\phi}\bra{0}_S = 0$.

\item Write the state $\phi = \proj{\phi}$ as
\[ \phi = \sum_{x,y} \phi_{x_1,\dots,y_n} \ket{x_1}\bra{y_1} \otimes \cdots \otimes \ket{x_n} \bra{y_n}.\]
Then, for any $S = \{i_1,\dots,i_k\}$,
\[ \phi_S = \sum_{x,y} [x_i = y_i, \forall i \in \bar{S}] \phi_{x_1,\dots,y_n} \ket{x_{i_1}}\bra{y_{i_1}} \otimes \cdots \otimes \ket{x_{i_k}} \bra{y_{i_k}}, \]
which implies
\[ \tr \proj{0}_S \ket{\phi}\bra{\phi}_S = \sum_x [x_i = 0, \forall i \in S] |\alpha_x|^2, \]
and hence, similarly to part (2),
\beas \sum_{S \subseteq [n]} \gamma^{|S|} \tr \proj{0}_S
\ket{\phi}\bra{\phi}_S
 &=& \sum_{x \neq 0}|\alpha_x|^2\sum_{k=0}^{n-|x|} \gamma^k \binom{n-|x|}{k}\\
 &=& (1+\gamma)^n \sum_{x \neq 0} |\alpha_x|^2 \left(\frac{1}{1+\gamma}\right)^{|x|}.
 \eeas

\item The last term can be trivially bounded using 
\[ |\tr \proj{\phi}_S(\ket{0}\bra{\phi}+\ket{\phi}\bra{0})_S| \le 2. \]
However, it is possible to get a better bound with a bit more work. We expand
\begin{eqnarray*}
\begin{aligned}
\sum_{S \subseteq [n]} \gamma^{|S|} &\tr \proj{\phi}_S \ket{0}\bra{\phi}_S =\\
\sum_{S \subseteq [n]} \gamma^{|S|} &\sum_{x,y,z} \alpha_x \alpha_y^* \alpha_z^* [z_i=0,i \in \bar{S}][x_i=y_i,i \in \bar{S}] \tr \ket{x_1}\braket{y_1}{0}\bra{z_1} \otimes \cdots \otimes \ket{x_n}\braket{y_n}{0}\bra{z_n}\\
&= \sum_{S \subseteq [n]} \gamma^{|S|}  \sum_{x,y,z} \alpha_x \alpha_y^* \alpha_z^* [z_i=0,i \in \bar{S}][x_i=y_i,i \in \bar{S}][y_i=0,i \in S][x_i=z_i,i \in S]\\
&= \sum_{|y\wedge z|=0} \alpha_{y \vee z} \alpha_y^* \alpha_z^* \sum_{S \subseteq [n]} \gamma^{|S|} [y_i=0,i \in S][z_i=0,i \in \bar{S}]\\
&= \sum_{|y\wedge z|=0} \alpha_{y \vee z} \alpha_y^* \alpha_z^* \gamma^{|z|} (1+\gamma)^{n-|y|-|z|}.
\end{aligned}
\end{eqnarray*}
This expression can be upper bounded as follows:
\bea
\lefteqn{\nonumber\sum_{|y\wedge z|=0} \alpha_{y \vee z} \alpha_y^* \alpha_z^* \gamma^{|z|} (1+\gamma)^{-(|y|+|z|)}}\\
\nonumber&\le& \sqrt{\sum_{|y\wedge z|=0} |\alpha_y|^2 |\alpha_z|^2} \sqrt{\sum_{|y\wedge z|=0} \frac{\gamma^{2|z|}}{(1+\gamma)^{2|y \vee z|}} |\alpha_{y \vee z}|^2}
\nn \\
\nonumber&\le& \left(\sum_x (1+\gamma)^{-2 |x|} |\alpha_x|^2 \left(\sum_{|y \wedge z|=0} \gamma^{2|z|}[y \vee z=x]\right)\right)^{1/2} \nn\\
&=& \left(\sum_x \left(\frac{1+\gamma^2}{(1+\gamma)^2}\right)^{|x|} |\alpha_x|^2\right)^{1/2}.
\label{eq:32-ub}
\eea

%This bound is not that far from optimal; it can be calculated explicitly that taking $\ket{\phi} = (\ket{01}+\ket{10}+\ket{11})/\sqrt{3}$ gives $\bbE_S[ \tr \proj{\phi}_S \ket{0}\bra{\phi}_S] = 1/(6\sqrt{3})$. Conjecture: this is actually the worst case.
\end{enumerate}

Combining these terms, we have
\begin{multline*}
\sum_{S \subseteq [n]} \gamma^{|S|} \tr\psi_S^2 \le (1+\gamma)^n ((1-\eps)^2+2\eps(1-\eps)\sum_{x \neq 0} |\alpha_x|^2(1+\gamma)^{-|x|}(\gamma^{|x|}+1) + \eps^2 +\\
4\eps^{3/2}\sqrt{1-\eps} \left(\sum_x \left(\frac{1+\gamma^2}{(1+\gamma)^2}\right)^{|x|} |\alpha_x|^2\right)^{1/2}).
\end{multline*}
%\sum_{S \subseteq [n]} \gamma^|S| \tr\psi_S^2 &\le& (1 - \eps)^2  + 4\eps(1-\eps)\sum_x 2^{-|x|} |\alpha_x|^2 + \eps^2 + 4 \eps^{3/2} \sqrt{1-\eps}\\
%&\le& 1 - 2 \eps + 4 \eps^{3/2} + 2 \eps^2 + 4 \eps \sum_x 2^{-|x|} |\alpha_x|^2,
%
Note that $(1+\gamma)^{-|x|}(\gamma^{|x|}+1)$ decreases with $|x|$ for all $\gamma > 0$, as does $(1+\gamma^2)^{|x|}(1+\gamma)^{-2|x|}$. To complete the proof, we will show that $\ket{\phi}$ has no weight 1 components (i.e.\ $\alpha_x = 0$ for $|x| < 2$).  In the contribution from \eq{32-ub}, this implies that only the $|x|\geq 4$ terms contribute (since $x=y\vee z$ and $y\wedge z=\emptyset$).
 Therefore, $\ket{\phi}$ having no weight 1 components would imply that
\[\sum_{S \subseteq [n]} \gamma^{|S|} \tr\psi_S^2 \le (1+\gamma)^n \left(1 - \frac{4\epsilon}{(1+\gamma)^2} \left(\gamma(1 - \eps) - \left(\frac{(1+\gamma^2)^2}{(1+\gamma)^2}\right)\eps^{1/2} \right)\right), \]
which would imply the theorem. Now, for any $\theta$, $\varphi$, we have $1-\eps \geq |(\cos\theta \bra 0 + e^{i\varphi} \sin\theta\bra 1)\ot \bra 0^{\ot n-1} \ket \psi|^2$. Picking $\theta$ such that
\[ \cos \theta = \frac{|\braket{0}{\psi}|}{\sqrt{|\braket{0}{\psi}|^2 + |\braket{10^{n-1}}{\psi}|^2}},\]
and $\varphi$ such that $e^{i \varphi} \braket{10^{n-1}}{\psi} > 0$, it is easy to see that
\[ 1-\eps \geq |\cos \theta \braket{0}{\psi} + e^{i \varphi} \sin \theta \braket{10^{n-1}}{\psi}|^2 = |\braket{0}{\psi}|^2 + |\braket{10^{n-1}}{\psi}|^2.\]
However, we have assumed that $1-\eps = |\braket{0}{\psi}|^2$, so this implies that $\braket{10^{n-1}}{\psi}=0$.  Repeating the argument for the other $n-1$ subsystems shows that $\ket\psi$ is indeed orthogonal to every state with Hamming weight at most 1, so $\ket{\phi}$ has no weight 1 components.
\end{proof}

%-------------------------------------------------------------------------------

% The ACM style seems to conflict with hyperref if we have references
% in section headings.
\section{Proof of \ifnotacm{\thmref{prodtest}: }\fi correctness of the product test}
\label{sec:prodtest}

In this appendix, we prove correctness of the product test (Theorem
\ref{thm:prodtest}). Let the test be defined as in
\protoref{prodtest}. The following lemma from \cite{qboolean}
expresses the probability of passing in terms of the partial traces of
the input states; we include a proof for completeness.

\setcounter{thm}{\value{test}}

\begin{lem}
Let $\Ptest(\rho,\sigma)$ denote the probability that the product test passes when applied to two mixed states $\rho,\sigma \in \mathcal{B}(\bbC^{d_1} \ot \cdots \ot \bbC^{d_n})$. Define $\Ptest(\rho) := \Ptest(\rho,\rho)$. Then
\[ \Ptest(\rho,\sigma) = \frac{1}{2^n} \sum_{S \subseteq [n]} \tr \rho_S \sigma_S, \]
and in particular
\[ \Ptest(\rho) = \frac{1}{2^n} \sum_{S \subseteq [n]} \tr \rho_S^2. \]
If $d_1 = d_2 = \dots = d_n = d$, for some $d$, then
\[ \Ptest(\rho) = \left( \frac{d+1}{2} \right)^n \tr (\mathcal{D}_{1/\sqrt{d+1}}^{\ot n}\, \rho)^2. \]
\end{lem}

Note that we can in fact assume that $d_1 = d_2 = \dots = d_n = d$ without
loss of generality by setting $d=\max(d_1,\ldots,d_n)$, and embedding
each of $\bbC^{d_1}, \ldots,\bbC^{d_n}$ into $\bbC^d$ in the natural
way.  This padding operation neither affects the probability
of the swap tests passing nor changes the distance to the closest
product state.

\begin{proof}
Let $\mathcal{F}$ denote the swap (or flip) operator that exchanges two quantum systems of equal but arbitrary dimension, with $\mathcal{F}_S$ denoting the operator that exchanges only the qudits in the set $S$. Then we have
\[
\Ptest(\rho,\sigma) = \tr (\rho \otimes \sigma) \left( \frac{I + \mathcal{F}}{2} \right)^{\otimes n} = \frac{1}{2^n} \sum_{S \subseteq [n]} \tr (\rho \otimes \sigma)\,\mathcal{F}_S = \frac{1}{2^n} \sum_{S \subseteq [n]} \tr \rho_S \sigma_S.
\]
The second part then follows from Lemma \ref{lem:depolarise}.
\end{proof}

We now analyse the probability of the product test passing for general $n$. We first note that, in the special case where $n=2$, it is possible to analyse the probability of passing quite tightly. The proof of the following result, which is implicit in previous work of Wei and Goldbart \cite{wei03}, is essentially immediate from Lemma \ref{lem:test}.

\setcounter{thm}{\value{skip}}
\begin{lem}
\label{lem:bipartite}
Let $\ket{\psi} \in \bbC^{d_1} \otimes \bbC^{d_2}$, where $d_1 \le d_2$, be a bipartite pure state with Schmidt coefficients $\sqrt{\lambda_1} \ge \sqrt{\lambda_2} \ge \dots \ge \sqrt{\lambda_{d_1}}$. Then
\[ \Ptest(\proj{\psi}) = \frac{1}{2}\left(1 + \sum_i \lambda_i^2 \right), \]
while
\[ 1 - \epsilon := \max_{\ket{\phi_1}, \ket{\phi_2}} |\braket{\psi}{\phi_1}\ket{\phi_2}|^2 = \lambda_1. \]
In particular,
\[ 1 - \epsilon + \frac{d_1}{2(d_1-1)} \epsilon^2 \le \Ptest(\proj{\psi}) \le 1 - \epsilon + \epsilon^2. \]
\end{lem}

\setcounter{skip}{\value{thm}}

We are finally ready to prove Theorem \ref{thm:prodtest}. The proof is split into two parts, which we formalise as separate theorems. The first part holds when $\epsilon$ is small, and depends on the results proven in Appendix \ref{sec:depolarising}. The second part holds when $\epsilon$ is large, and is proved using the first part.

\setcounter{thm}{\value{prodtestsmalleps}}
\begin{thm}
Given $\ket{\psi}\in \bbC^{d_1} \ot \cdots \ot \bbC^{d_n}$, let
\[ 1 - \eps = \max\{ |\braket{\psi}{\phi_1,\ldots,\phi_n}|^2 : \ket{\phi_i}\in\bbC^{d_i}, 1 \le i \le n\}.\]
Then 
\[ 1 - 2\eps + \eps^2 \le \Ptest(\proj{\psi}) \le 1 - \eps + \eps^2 + \eps^{3/2}.\]
\end{thm}

\begin{proof}
The lower bound holds by general arguments. It is immediate that, if applied to $\ket{\phi_1,\ldots,\phi_n}$, the product test succeeds with probability 1. As the test acts on two copies of $\ket{\psi}$, which has overlap $1-\epsilon$ with $\ket{\phi_1,\ldots,\phi_n}$, it must succeed when applied to $\ket{\psi}$ with probability at least $(1-\epsilon)^2$. The upper bound follows from Lemma \ref{lem:test} and Theorem \ref{thm:stability}. The statement of Theorem \ref{thm:stability} only explicitly covers the case where the dimensions of all the subsystems are the same; however, as noted above, we can assume this without loss of generality.
%however, it is clear from inspecting the proof that it is sufficient to give the claimed bound on $\Ptest(\proj{\psi})$ for arbitrary local dimensions.
\end{proof}

This result is close to optimal. At the low end, the state $\ket{\psi} = \sqrt{1-\eps}\ket{0^n} + \sqrt{\eps}\ket{1^n}$ has $\Ptest(\proj{\psi}) = 1 - 2 \eps + 2\eps^2 + o(1)$. At the high end, for $\ket{\psi} = \sqrt{1-\eps}\ket{00} + \sqrt{\eps}\ket{11}$, $\Ptest(\proj{\psi}) = 1 - \eps + \eps^2$. We also note that this result does not extend to a test for separability of mixed states; the maximally mixed state on $n$ qudits is separable but it is easy to verify that $\Ptest(I/d^n) = ((d+1)/2d)^n$, which approaches zero for large $n$.

Theorem \ref{thm:prodtestsmalleps} only gives a non-trivial upper bound on the probability of passing when $\epsilon$ is small (up to $\epsilon = \frac{1}{2}(3 - \sqrt{5}) \approx 0.38$). We now show that the product test also works in the case where the state under consideration is far from any product state. We will need two lemmas.

\setcounter{thm}{\value{skip}}
\begin{lem}
\label{lem:pproduct}
Given $\ket{\psi}\in \bbC^{d_1} \ot \cdots \ot \bbC^{d_n}$, let $\Ptest^P(\proj{\psi})$ be the probability that the $P$-product test -- the test for being product across partition $P$ -- passes. Then, for all $P$, $\Ptest^P(\proj{\psi}) \le \Ptest(\proj{\psi})$.
\end{lem}

\begin{proof}
The subspace corresponding to the usual product test passing is contained within the subspace corresponding to the $P$-product test passing.
\end{proof}

\begin{lem}
\label{lem:close}
Let $\ket{\psi}$, $\ket{\phi}$ be pure states such that
$|\braket{\psi}{\phi}|^2 = 1 - \epsilon$, and let $P$ satisfy $0\leq
P\leq I$; e.g. $P$ might be a projector. Then
$|\langle\psi|P|\psi\rangle - \langle\phi|P|\phi\rangle| \le
\sqrt{\epsilon}$. 
\end{lem}

\begin{proof}
We can directly calculate $\frac{1}{2}\|\, \proj \psi - \proj\phi\,\|_1 = \sqrt{\eps}$.
This then gives the claimed upper bound on $| \tr P( \proj \psi - \proj \phi)|$ (see \cite[Chapter 9]{nielsen00}).
\end{proof}

\setcounter{skip}{\value{thm}}

\setcounter{thm}{\value{prodtestlargeeps}}

\begin{thm}
Given $\ket{\psi}\in \bbC^{d_1} \ot \cdots \ot \bbC^{d_n}$, let
\[ 1 - \eps = \max\{ |\braket{\psi}{\phi_1,\ldots,\phi_n}|^2 : \ket{\phi_i}\in\bbC^{d_i}, 1 \le i \le n\}.\]
Then, if $\eps \ge 11/32 > 0.343$, $\Ptest(\proj{\psi}) \le 501/512 < 0.979$.
\end{thm}

\begin{proof}
For simplicity, in the proof we will use a quadratic upper bound on $\Ptest(\proj{\psi})$ that follows by elementary methods from Theorem \ref{thm:prodtest}: $\Ptest(\proj{\psi}) \le 1 - \frac{3}{4} \epsilon + 2 \epsilon^2$. For a contradiction, assume that $\Ptest(\proj{\psi}) > p := 501/512$, while $\epsilon \ge 11/32$.

For any partition $P$ of $[n]$ into $1 \le k \le n$ parts, let $\ket{\phi_P}$ be the product state (with respect to $P$) that maximises $|\braket{\psi}{\phi}|^2$ over all product states $\ket{\phi}$ (with respect to $P$). If
\[ 1 - h \le |\braket{\psi}{\phi_P}|^2 \le 1 - \ell, \]
where for readability we define $\ell := 1/32$ and $h := 11/32$, then by the quadratic bound given above the $P$-product test passes with probability $\Ptest^P(\proj{\psi}) \le p$, implying by Lemma \ref{lem:pproduct} that $\Ptest(\proj{\psi}) \le p$. Therefore, as we are assuming that $\ket{\psi}$ is a counterexample to the present theorem, there exists a $k$ such that $|\braket{\psi}{\phi}|^2 > 1 - \ell$ for some $\ket{\phi}$ that is product across $k$ parties, and yet $|\braket{\psi}{\phi}|^2 < 1 - h$ for all $\ket{\phi}$ that are product across $k+1$ parties.

So, for this $k$, let $\ket{\phi_1} \cdots \ket{\phi_k}$ be a state
that maximises $|\braket{\psi}{\phi_1,\ldots,\phi_k}|^2$.  Thus there
is some $\eps' < \ell$ such that we can write $\ket\psi$ as
\[ \ket{\psi} = \sqrt{1 - \epsilon'} \ket{\phi_1} \cdots \ket{\phi_k} + \sqrt{\epsilon'} \ket{\xi}. \]
If $k=1$, then trivially $\ket{\phi_1} = \ket{\psi}$ and $\epsilon'=0$.
%For $k > 1$, by the arguments at the end of \thmref{stability}, the $i^{\text{th}}$ marginal of $\ket\xi$ has support orthogonal to $\ket{\phi_i}$.
Assume without loss of generality that $\ket{\phi_1}$ is a state of
two or more qudits.
Now we know that
\be  \max_{\ket{\phi'_{1,1}},\ket{\phi'_{1,2}}}
|\braket{\phi_1}{\phi'_{1,1}}\ket{\phi'_{1,2}}|^2 
(1-\epsilon') < 1 - h, 
\label{eq:phi1-ent}\ee
or $\ket{\phi'_{1,1}}\ket{\phi'_{1,2}}\ket{\phi_2} \cdots \ket{\phi_k}$ would
be a $(k+1)$-partite state with overlap at least $1 - h$ with
$\ket{\psi}$.   (Here we have used the fact that for $k > 1$, by the arguments at the end of \thmref{stability}, $\ket\xi$ is
orthogonal to $\ket{\phi'_{1,1}}\ket{\phi'_{1,2}}\ket{\phi_2} \cdots
\ket{\phi_k}$ for any choice of 
$\ket{\phi'_{1,1}}$, $\ket{\phi'_{1,2}}$.)
Let $1-\delta = 
 \max_{\ket{\phi'_{1,1}},\ket{\phi'_{1,2}}}
|\braket{\phi_1}{\phi'_{1,1}}\ket{\phi'_{1,2}}|^2$.  Then
\eq{phi1-ent} implies that
$$1-\delta < \frac{1-h}{1-\eps'} < \frac{1-h}{1-\ell} =
\frac{21}{31}.$$
Using the exact expression given in \lemref{bipartite}, we find that $\Ptest(\proj{\phi_1}) < 751/961$ (if $10/31<\delta\le 1/2$, this follows from $\Ptest(\proj{\phi_1}) \le 1+\delta+\delta^2$; if $\delta \ge 1/2$, then $\Ptest(\proj{\phi_1}) \le 3/4$ always). Next we use Lemma \ref{lem:close} to obtain
\bas \Ptest(\proj{\psi}) &\le \Ptest(\proj{\phi_1}\ot \cdots \ot
\proj{\phi_k}) + \sqrt{\eps'}\\
& < \Ptest(\proj{\phi_1}) + \sqrt{\ell} \\ &
< \frac{751}{961} + \sqrt{\frac{1}{32}} < 0.96. \eas
But we previously assumed that $\Ptest(\proj{\psi}) > p > 0.978$. We
have reached a contradiction, so the proof is complete.
\end{proof}

\setcounter{thm}{\value{skip}}

One might hope that this theorem could be improved to show that, as $\epsilon \rightarrow 1$, $\Ptest(\proj{\psi})$ necessarily approaches $0$. However, this is not possible. Consider the $d \times d$-dimensional bipartite state $\ket{\Phi} = \frac{1}{\sqrt{d}} \sum_{i=1}^d \ket{ii}$. It is easy to verify using Lemma \ref{lem:bipartite} that $\Ptest(\proj{\Phi})= 1/2(1+1/d)$ while $\max_{\ket{\phi_1}, \ket{\phi_2}} |\braket{\Phi}{\phi_1}\ket{\phi_2}|^2 = 1/d$.

Combining Theorems \ref{thm:prodtestsmalleps} and \ref{thm:prodtestlargeeps}, we obtain Theorem \ref{thm:prodtest} and thus have proven correctness of the product test. The constants in Theorem \ref{thm:prodtestlargeeps} have not been optimised as far as possible and could be improved somewhat.

%-------------------------------------------------------------------------------

\section{Classes of measurement operators}
\label{sec:measure-defs}
In this appendix, we define the classes of measurement operators used
in our paper and other relevant literature on $\QMA(2)$, such as
\cite{brandao11}. Our definitions mostly follow the conventions of
quantum information theory. Each class of measurement operators
describes operators on $\bbC^d \ot \bbC^d$.

\begin{itemize}
\item BELL is the set of $M$ that can be expressed as
\be M = \sum_{(i,j)\in S} \alpha_i \ot \beta_j,\ee
where $\sum_i \alpha_i = I$ and $\sum_j \beta_j = I$, and $S$ is a set
of pairs of indices.  In other words, the systems are locally
measured, obtaining outcomes $i$ and $j$, and then the verifier
accepts if $(i,j)\in S$.
\item LOCC$_1$ is the set of $M$ that can be realised by measuring the
  first system and then choosing a measurement on the second system
  conditional on the outcome of the first measurement.  Such $M$ can
  be written as
\be M = \sum_i \alpha_i \ot M_i, \ee
where $\sum_i \alpha_i = I$ and $0\leq M_i \leq I$ for each $i$.
\item LOCC is the set of $M$ that can be realised by alternating
partial  measurements on the two systems a finite number of times,
choosing each measurement conditioned on the previous outcomes.  An
inductive definition is that M is in LOCC if there exist operators
$\{E_i\}, \{M_i\}$, with $\sum_i E_i\leq I$ and each $M_i\in{\rm LOCC}$,
such that either $M = \sum_i (\sqrt{E_i}\ot I)M_i(\sqrt{E_i}\ot I)$ or
$M = \sum_i (I\ot \sqrt{E_i})M_i(I\ot \sqrt{E_i})$.   For the base
case, it suffices to take $I\in {\rm LOCC}$.
\item $\SEP$ is the set of $M$ such that
\be M = \sum_i \alpha_i \ot \beta_i
\label{eq:sep-def}\ee
for some positive semidefinite (WLOG rank one) matrices $\{\alpha_i\},\{\beta_i\}$.
(Note: other works define $\SEP$ to be the smaller set of $M$ for which both $M$ and $I-M$ can be decomposed as in \eq{sep-def}, and use the term $\SEP_{\rm YES}$ to describe the measurements for which only $M$ has to satisfy \eq{sep-def}.)
\item SEP-BOTH is the set of $M$ for which $M\in\SEP$ and $I-M\in\SEP$.
\item PPT (positive partial transpose~\cite{PPT1,PPT2})
 is the set of $M$ for which $M^\Gamma\geq 0$, where $\Gamma$
  is the partial transpose map defined by 
$(\ket i\bra j \ot \ket  k\bra l)^\Gamma = 
(\ket i\bra j \ot \ket  l\bra k)$.
Again note that this definition does not require $I-M\in\PPT$.
\item PPT-BOTH is the set of $M$ for which $M\in\PPT$ and $I-M\in\PPT$.
\item ALL has no restrictions on $M$ other than $0\leq M\leq I$.
\end{itemize}
We note that SEP-BOTH and PPT-BOTH are natural relaxations of LOCC
because they preserve the property that both $M$ and $I-M$ must be
realisable through local operations and classical communication.  On
the other hand, SEP and PPT are more natural when we consider $M$ by
itself and do not wish to consider additional constraints on $I-M$.

These sets satisfy the following inclusions, all of which are known to
be strict
$$
\begin{array}{ccccccccccc}
{\rm BELL} &\subset& {\rm LOCC}_1 &\subset& {\rm LOCC}& \subset &
\text{SEP-BOTH} & \subset & \text{PPT-BOTH} \\
&&&&&&\cap & & \cap \\
&&&&&&{\rm SEP} & \subset & \text{PPT} & \subset & {\rm ALL}
\end{array}
$$

\section{Nonexistence of an LOCC product test}
\label{sec:locc}

A natural extension of the idea of product state testing is to a distributed setting where two parties, each of whom receives one copy of an $n$-partite state $\ket{\psi}$, must determine whether $\ket{\psi}$ is product using only local operations and classical communication (LOCC). Indeed, following the completion of an initial version of this work, it was shown by Brand\~ao, Christandl and Yard that, if there were an efficient LOCC protocol for product state testing, then $\QMA(k) = \QMA$ \cite{brandao11,brandao10a}.

In this appendix, we show that unfortunately no such LOCC protocol exists.
In fact, we rule out the larger class of PPT-BOTH 
measurements (defined in Appendix~\ref{sec:measure-defs}).
 Our impossibility result holds for the easiest version
of this task, in which $n=2$. 
For simplicity, here we only consider the case where the test uses 2 copies of $\ket{\psi}$;
one can show a similar result when the number of copies is larger but the proof is
significantly more complicated~\cite{Harrow-perm}.

Formally, we define a product test as a measurement $\{M,I-M\}$ that
acts on $\ket{\psi}^{\ot 2} = \ket{\psi}^{A_1B_1} \ot \ket{\psi}^{A_2B_2}$ with
outcome $M$ corresponding to ``product'' and $I-M$ corresponding to
``not product.''  There is no good canonical way to express the
validity of a product test.  One rather general way we might do this
is to say that there are functions $f(\eps)$ and $g(\eps)$ such that
if $\ket{\psi}$ has overlap $1-\eps$ with the closest product state
then its probabity $\Ptest$ of passing the product test satisfies
\be f(\eps) \leq \Ptest \leq g(\eps).
\label{eq:ptest-bounds}\ee  
For example,  \thmref{prodtest} shows that our product test satisfies
\eq{ptest-bounds} with $f(\eps)=1-c_1\eps$ and $g(\eps)=1-c_2\eps$
with $c_1 > c_2 > 0$.

For our impossibility result, we will use a different and simpler
success measure.  Define the completeness $c$ of a product test to be
the average probability of accepting a random product state
$\ket{\psi} = \ket{\psi_A} \ot \ket{\psi_B}$, and define the soundness
$s$ to be the average probability of accepting a random bipartite
state $\ket{\psi}$.  While strictly speaking a random bipartite state
may sometimes be close to a product state, it has an overwhelmingly
high probability of being close to maximally entangled.  Thus, this
demand for a product test is nearly as undemanding as possible. Finally,
define the bias of the test as $b = c-s$.

\begin{thm}
Any 2-copy PPT-BOTH product test for bipartite $d \times
d$-dimensional product states has bias which is $O(1/d)$. 
\end{thm}

\begin{proof}
Let $\ket{\psi}$ be a bipartite $d \times d$-dimensional state on the system AB. Imagine we have a protocol which takes as input two copies of $\ket{\psi}$, written $\ket{\psi}_1 \ket{\psi}_2$, and attempts to determine whether $\ket{\psi}$ is product across systems A and B. Consider two distributions $\mathcal{D}_0$, $\mathcal{D}_1$ on bipartite $d \times d$ states. Let $M$ be a measurement operator which accepts states drawn from $\mathcal{D}_1$ with probability at least $c$, and rejects states from $\mathcal{D}_0$ with probability at least $s$. Then
\[ \bbE_{\psi \sim \mathcal{D}_1} \bbE_{\phi \sim \mathcal{D}_0} \tr M(\psi \ot \psi - \phi \ot \phi) \ge c - s, \]
implying
\[ \tr M( \bbE_{\psi \sim \mathcal{D}_1}(\psi \ot \psi) - \bbE_{\phi \sim \mathcal{D}_0} (\phi \ot \phi)) \ge c - s. \]
Taking $\mathcal{D}_1$ to be the uniform distribution over product states, via a standard calculation we obtain
\[ \bbE_{\psi \sim \mathcal{D}_1}(\psi \ot \psi) = \bbE_{\psi} \bbE_{\phi} (\psi_A \ot \phi_B)^{\ot 2} = \left(\frac{1}{d(d+1)}(I + \mathcal{F})_{A_1 A_2}\right) \ot \left(\frac{1}{d(d+1)}(I + \mathcal{F})_{B_1 B_2}\right), \]
where as elsewhere $\mathcal{F}$ is the swap operator. Now let $\mathcal{D}_0$ be the uniform distribution on bipartite $d \times d$ states. In this case we have
\[ \bbE_{\phi \sim \mathcal{D}_0}(\phi \ot \phi) = \frac{1}{d^2(d^2+1)}(I + \mathcal{F})_{12}. \]
Thus
\beas
\Delta &:=& \bbE_{\psi \sim \mathcal{D}_1}(\psi \ot \psi) - \bbE_{\phi \sim \mathcal{D}_0}(\phi \ot \phi)\\
&=& \frac{1}{d^2(d+1)^2}\left(I_{A_1 A_2} \ot \mathcal{F}_{B_1 B_2} + \mathcal{F}_{A_1 A_2} \ot I_{B_1 B_2} \right) - \frac{2}{d(d^2+1)(d+1)^2} \left( I + \mathcal{F}_{12} \right).
\eeas
Assume that $M$ is PPT across the 1:2 split. We want to maximise $\tr M \Delta$ assuming that $-I \le M \le I$ and $-I \le M^\Gamma \le I$, where the second is the PPT constraint. Further, as distributions $\mathcal{D}_0$ and $\mathcal{D}_1$ are invariant under product unitaries, we can without loss of generality assume that $M$ commutes with $U_{A_1} \ot U_{A_2} \ot V_{B_1} \ot V_{B_2}$ for all unitaries $U$ and $V$, which implies that
\[ M = w I + x(\mathcal{F}_{A_1 A_2} \ot I_{B_1 B_2}) + y (I_{A_1 A_2} \ot \mathcal{F}_{B_1 B_2}) + z \mathcal{F}_{12} \]
for some $w$, $x$, $y$, $z$. By direct calculation
\beas
\tr M \Delta &=& \frac{1}{d^2(d+1)^2}\left( 2d^3 w + (d^2+d^4)x + (d^2+d^4)y + 2d^3 z \right) + O(1/d)\\
&=& x + y + O(1/d).
\eeas
On the other hand,
\[ M^\Gamma = w I + x d(\Phi_{A_1 A_2} \ot I_{B_1 B_2}) + y d(I_{A_1 A_2} \ot \Phi_{B_1 B_2}) + z d^2 (\Phi_{A_1 A_2} \ot \Phi_{B_1 B_2}), \]
where $\ket{\Phi} = \frac{1}{\sqrt{d}} \sum_{i=1}^d \ket{ii}$. So $|x + y| = O(1/d)$, and we are done.
\end{proof}

%-------------------------------------------------------------------------------

\section{Proof of correctness of the protocol to put \texorpdfstring{$\QMA(k)$}{QMA(k)} in \texorpdfstring{$\QMA(2)$}{QMA(2)}}
\label{sec:protoproof}

This section proves several of the claims made in \secref{QMA}.  First we prove \lemref{qmak-sim} by showing the validity of \protoref{qmak}.

\begin{replem}{lem:qmak-sim}
\lemqmaksim
\end{replem}

\begin{proof}
  It is obvious that this protocol achieves completeness $(1+c)/2$: if the
  Merlins follow the protocol, the product test passes with certainty,
  so Arthur either accepts with probability 1, or with
  the same probability that $\mathcal{A}$ accepts, which is at least $c$.
  Showing soundness is somewhat more complicated.

%, and we will need the following ``gentle measurement'' lemma.
% \begin{lem}[Gentle measurement lemma \cite{winter99,ogawa02}]
% \label{lem:gentle}
% Let $\rho$ be a density operator, and let $0 \le X \le I$ be a
% projector such that $\tr \rho X \ge 1 - \delta$. Then $\frac{1}{2}\|
% \rho - X \rho X \|_1 \le \sqrt{\delta}$. 
% \end{lem}

Assume that Arthur receives states $\ket{\phi_1}$ and $\ket{\phi_2}$
such the maximal overlap of $\ket{\phi_1}$ (resp.\ $\ket{\phi_2}$)
with a product state is $1-\epsilon_1$ (resp.\ $1-\epsilon_2$), and
set $\epsilon = \frac{1}{2}(\epsilon_1 + \epsilon_2)$. Further assume
that the product test would accept $\ket{\phi_1}^{\otimes 2}$ (resp.\
$\ket{\phi_2}^{\otimes 2}$) with probability $1-\delta_1$ (resp.\
$1-\delta_2$). 

Let $1-\delta$ be the probability that the product test would accept
$\ket{\phi_1}\ot \ket{\phi_2}$. We first show that this can be upper
bounded in terms of the probabilities of accepting $\ket{\phi_1}^{\ot 2}$ and
$\ket{\phi_2}^{\ot 2}$. The probability that the product test accepts is 
\beas
\frac{1}{2^k} \sum_{S \subseteq [k]} \tr (\phi_1)_S (\phi_2)_S & \leq & 
\frac{1}{2^k} \sum_{S \subseteq [k]} \sqrt{\tr (\phi_1)_S^2} \sqrt{\tr (\phi_2)_S^2}\\
&\leq&  \frac{1}{2^k} \sum_{S \subseteq [k]}
\frac{ \tr (\phi_1)_S^2 + \tr (\phi_2)_S^2}{2}
\\
&=& \frac{1}{2} \left( \Ptest(\phi_1) + \Ptest(\phi_2) \right)\\
&=& 1 - \frac{1}{2} \left( \delta_1 + \delta_2 \right).
\eeas
Thus we have the bound from \thmref{prodtest} that $\delta
\geq \frac{11}{512}\epsilon$.  
On the other hand, if Arther chooses to measure $M$, then by
\lemref{close} his probability of accepting is $\leq s
+ \frac{\sqrt{\eps_1}+\sqrt{\eps_2}}{2} \leq s + \sqrt{\eps}$.
 Combining the two tests, we find that the acceptance
probability is
\ba s' \leq \max_{\eps \leq \frac{512}{11}\delta}
\frac{1-\delta + \min(1,s + \sqrt{\eps})}{2}
% \nn\\ & \leq \max_{\delta}\min(1-\delta, s +
%\sqrt{\frac{512}{11}\delta} + \sqrt{\delta})\nn\\ 
 \leq 1 - \frac{(1-s)^2}{100}.
\label{eq:soundness-loss}\ea
To obtain the last inequality, we observe that the worst case is
obtained when $\sqrt{\eps}=1-s=\sqrt{\frac{512}{11}\delta}$.
% Here is the calculation:
% Rewrite as max(delta, gamma + A sqrt(delta)), 
% where A := sqrt(2*512/11)+1, gamma = 1-s
% The max is achieved when both sides are equal.
% Quadratic equation yields sqrt(delta) = sqrt(A^2/4 + gamma) - A/2
% Taylor expansion gives sqrt(delta) > gamma/A (1 - gamma/A^2).
% gamma < 1 gives delta> (1/A^2)(1-1/A^2)^2 gamma^2 > gamma^2/114

As a result of \eq{soundness-loss}, a $k$-prover soundness-$s$ protocol can be
simulated by a 2-prover protocol with soundness $s'$.  If $k\leq
\poly(n)$, then the messages will still have a polynomial number of
qubits.
\end{proof}

\ifacm
\begin{remark}[]
\else
\begin{remark}
\fi
Our choice of protocol in \protoref{qmak} was carefully designed to yield a separable protocol.  One subtlety in doing so is that measurements in SEP do not compose the same way that separable operations do.  Indeed, if $A,B\in \SEP$, then it does {\em not} follow that $A^{1/2} B A^{1/2}\in \SEP$, even in the case when $A$ is a projector.  For example, let $A$ project onto the symmetric subspace and let $B = \proj{0,1}$.  On the other hand, other choices of protocol could also yield separable measurements.  For example, $(M \ot I)P(M\ot I)$ would work (after using \lemref{completeness} to amplify completeness), where $P$ denotes the product test, $M$ is a measurement on Alice's $k$ systems and $I$ acts on Bob's $k$ systems.  We are grateful to Fernando Brand\~ao for pointing out this issue to us.
\end{remark}

\begin{comment}
 The astute reader might wonder if many of the complexity-theoretic
 results in this paper might follow more simply from \lemref{qma2-sim}.
Indeed, \lemref{qma2-sim} implies that $\QMA_m(2)_{s,c} \subseteq
\QMA_{2m}^{\SEP}(2)_{\frac{3+s}{4},c}$, which together with
\lemref{qma-sep-ampl} and the results of
\cite{kobayashi03,Unentanglement}, can be used to obtain an alternate
proof of \thmref{qma2k}. However, this route does not appear to
suffice to obtain tight enough soundness guarantees for the hardness
proofs in Section \ref{sec:complex}.  More precisely, using these prior
results directly would only imply a $\QMA(2)$ protocol for $\sat$ with
$\sqrt{n} \polylog n$ qubits and soundness
$1-\Omega(1/\sqrt{n}\polylog n)$, rather than constant.
\end{comment}

Next, we prove \lemref{qma-sep-ampl}.
\begin{replem}{lem:qma-sep-ampl}
\lemqmasepampl
\end{replem}

\begin{proof}
In the original protocol, Arthur performs a measurement $M$ on $k$ states, each comprising $m$ qubits.  If the input is a YES instance, then there exists a product state on which $M$ has expectation value $\geq c$, whereas if the input is a NO instance, then for all product states, $M$ has expectation value $\leq s$.

For the modified protocol, each of the $k$ provers submits $\ell m$ qubits, and Arthur's measurement is $M^{\ot \ell}$.  If the input is a YES instance, then the provers can submit $\ell$ copies of the optimal input to the original protocol.  This state is still product (across the $k$ provers) and has probability $\geq c^\ell$ of being accepted.

The more interesting case is when the input is a NO instance.  In general, the provers can submit states that are entangled across the $\ell$ different parts of their message.   However, there cannot be any entanglement between the $k$ different provers.

 For the new protocol, we can imagine Arthur sequentially performing the measurement $\{M, I-M\}$ $\ell$ times and accepting only if the outcome is $M$ each time.  Since the input is a NO instance, if this measurement is applied to a state, pure or mixed, that is separable across the $k$ parties, then outcome $M$ will occur with probability $\leq s$.  Therefore the probability that all $\ell$ measurements have outcome $M$ will be $\leq s^\ell$ as long as conditioning on outcome $M$ does not induce any entanglement in the unmeasured states.  We need not consider the state that remains after outcome $I-M$, since Arthur rejects immediately when this outcome occurs.
 
 The final step of the proof is to show that applying the measurement $\{M,I-M\}$ to the first system of a multipartite product state and obtaining outcome $M$ will not create any entanglement across the $k$ provers. Suppose that the $i^{\text{th}}$ Merlin supplies the $\ell m$-qubit state $\ket{\varphi^{(i)}}^{P_1^i\ldots P_\ell^i}$, where each $P_j^i$ is an $m$-qubit system.  The original measurement $M$ is separable, so can be written as
$$M = \sum_j a_j \proj{\alpha_j^{(1)}} \ot \cdots \ot \proj{\alpha_j^{(k)}},$$
for $a_j\geq 0$ and $\ket{\alpha_j^{(i)}}$ unit vectors on $m$ qubits.  This measurement is applied sequentially to $P_1^{1,\ldots,k}$, then $P_2^{1,\ldots,k}$, and so on until $P_\ell^{1,\ldots,k}$.   (Here we write $P_j^{1,\ldots,k}$ as shorthand for $P_j^1\cdots P_j^k$.) When the measurement is applied to $P_1^{1, \ldots, k}$ and the outcome is $M$, the residual state is proportional to
\begin{multline*}
\tr_{P_1^{1,\ldots,k}} M^{P_1^{1, \ldots,k}} 
\L( \proj{\varphi^{(1)}}^{P_{1,\ldots,\ell}^1} \otimes \cdots \ot
\proj{\varphi^{(k)}}^{P_{1,\ldots,\ell}^k}\R)
\\ = \sum_j a_j \bigotimes_{i=1}^k \bra{\alpha_j^{(1)}}^{P_1^i} \ot I^{P_{2,\ldots,\ell}^i}
\proj{\varphi^{(i)}}^{P_{1,\ldots,\ell}^i} 
\ket{\alpha_j^{(1)}}^{P_1^i} \ot I^{P_{2.\ldots,\ell}^i},
\end{multline*}
 which is separable across the $P_{2,\ldots,\ell}^1:P_{2,\ldots,\ell}^2:\cdots:P_{2,\ldots,\ell}^k$ cut.
Therefore, by induction, the state always remain separable as long as outcome $M$ always occurs.  
 \end{proof}

For the reader's convenience, we briefly summarise here the proof of \lemref{completeness} (originally due to \cite{kobayashi03,Unentanglement}).
\begin{replem}{lem:completeness}
\lemcompleteness
\end{replem}
\begin{proof}
 The idea is to repeat the basic protocol $\ell$ times, and accept if there are $\geq \frac{c+s}{2}\ell$ ``accept'' outcomes or reject otherwise.  For YES instances, the provers can send $\ell$ copies of the same proofs, each of which will be accepted with probability $\geq c$.  Then a Chernoff bound yields that the completeness is $\geq 1-\exp(-\ell(c-s)^2/2)$.  For NO instances, each of the $\ell$ copies has probability $\leq s$ of being accepted.  Since the provers may submit entangled states, we can no longer guarantee that these events are independent, but still Markov's inequality implies that the probability of $\geq \frac{c+s}{2}\ell$ accept outcomes is $\leq \frac{2s}{c+s} = \frac{1}{1+(c-s)/2} \leq 1 - \frac{c-s}{3}$.  This last step uses the fact that $1/(1+x/2) \leq 1-x/3$, for $0\leq x \leq 1$.
\end{proof}

We now complete the proof of \thmref{qma2k}.
\begin{repthm}{thm:qma2k}
\thmqma2k
\end{repthm}
\begin{proof}
For the case of perfect completeness, we use \lemref{qmak-sim-sep} to replace the original $k$-prover protocol with a 2-prover protocol that still has perfect completeness, still has $1-1/q(n)$ soundness (where $q(n)$ is still a  polynomial of $n$) and now has a separable measurement performed by the verifier.  Now we simply repeat (using \lemref{qma-sep-ampl}) $q(n)\cdot p(n)$ times, and obtain soundness $\leq (1-1/q(n))^{q(n)p(n)} \leq e^{-p(n)}$.

When we merely have $c-s > 1/q(n)$, for some polynomial $q(n)$, then we first have to apply \lemref{completeness} with $\ell = p(n) + \log(243q(n)^2p(n))$  to replace the soundness with $1-1/3q(n)$ and the completeness with $1-\frac{\exp(-p(n))}{243q(n)^2p(n)}$.  Next, we apply \lemref{qmak-sim-sep} to leave the completeness the same, reduce the number of provers to 2, guarantee the measurement is in $\SEP$ and replace the soundness with $1-1/243q(n)^2$.  Finally, we repeat $243q(n)^2p(n)$ times and obtain soundness $\exp(-p(n))$ and completeness $1-\exp(-p(n))$.
\end{proof}

\section{Proof of correctness of the product unitary test}
\label{sec:unitary-proof}

This appendix is devoted to the proof of \thmref{unitary-test}.
In order to analyse the product unitary test in
\protoref{produnitarytest}, we will need to relate the maximum
overlap of an $n$-qudit unitary with a product operator
to the maximum overlap of that unitary with a product unitary.

\begin{lem}
\label{lem:nearestunitary}
Given $U\in U(d^n)$, let
\[ 1 - \eps = \max\{ \left|\hsip{U}{A_1 \otimes \dots \otimes A_n}\right|^2 : A_i \in M(d), \hsip{A_i}{A_i} = 1, 1 \le i \le n\}.\]
Then, if $\epsilon \le 1/2$, there exist $V_1,\dots,V_n \in U(d)$ such that $\left|\hsip{U}{V_1 \otimes \dots \otimes V_n}\right|^2 \ge (1-2 \epsilon)^2$.
\end{lem}

\begin{proof}
For all $1 \le i \le n$, let the polar decomposition of $A_i$ be $|A_i| C_i$, where $|A_i| = \sqrt{A_i A_i^\dag}$ and $C_i \in U(d)$. Set $A = \bigotimes_{i=1}^n A_i$, $C = \bigotimes_{i=1}^n C_i$. Then
\[ \hsip{C}{A} = \frac{1}{d^n} \prod_{i=1}^n \tr C_i^\dag |A_i| C_i = \frac{1}{d^n} \prod_{i=1}^n \tr |A_i| = \frac{1}{d^n} \max_{V \in U(d^n)} | \tr V A | \ge \sqrt{1 - \epsilon}. \]
This implies that we can expand
\[ U = \sqrt{1-\epsilon}\,A + D, \; \; C = \sqrt{1-\epsilon'}\,A + E \]
for some $\epsilon' \le \epsilon$ and matrices $D, E$ such that $\hsip{D}{D}=\epsilon$, $\hsip{E}{E}=\epsilon'$, $\hsip{A}{D} = 0$, $\hsip{A}{E} = 0$. So
\[ |\hsip{U}{C}| = | \sqrt{1-\epsilon} \sqrt{1-\epsilon'} + \hsip{D}{E} | \ge | \sqrt{1-\epsilon} \sqrt{1-\epsilon'} - \sqrt{\epsilon}\sqrt{\epsilon'}| \ge 1 - 2 \epsilon, \]
for $\epsilon \le 1/2$. This implies the lemma.
\end{proof}

We are now ready to prove correctness of the product unitary test. 
\begin{repthm}{thm:unitary-test}
\thmunitarytest
\end{repthm}
\begin{proof}
By the Choi-Jamio\l{}kowski isomorphism, there is a direct correspondence
between operators $M \in M(d)$ with $|\hsip{M}{M}|=1$ and normalised quantum states $\ket{v(M)}$. If we define
\[ 1 - \eps' := \max\{ \left|\hsip{U}{A_1 \otimes \dots \otimes A_n}\right|^2 : A_i \in M(d), \hsip{A_i}{A_i} = 1, 1 \le i \le n\}, \]
then by Theorem \ref{thm:prodtest}, if $\eps' \lesssim 0.0265$, $\Ptest(U) \le 1 - \eps' + \eps'^2 + \eps'^{3/2}$, and if $\eps' \gtrsim 0.0265$, $\Ptest(U) \le 501/512$. If $\eps' \ge 1/2$, then the result follows immediately. On the other hand, by Lemma \ref{lem:nearestunitary}, if $\eps' \le 1/2$, there exist $V_1,\dots,V_n \in U(d)$ such that $\left|\hsip{U}{V_1 \otimes \dots \otimes V_n}\right|^2 \ge (1-2 \epsilon')^2 \ge 1 - 4 \epsilon'$. Thus we have $\frac{1}{4} \epsilon \le \epsilon' \le \epsilon$. The theorem follows by combining the bound on $\eps$ and the bound on $\Ptest(U)$.
\end{proof}

%-------------------------------------------------------------------------------

\section{Interpretation of the product test as an average over product states}
\label{sec:average}

We have seen (via Lemma \ref{lem:test}) that the probability of the
product test passing when applied to some state $\ket{\psi} \in
(\bbC^d)^{\ot n}$ is equal to the average purity, across all choices
of subsystem $S \subseteq [n]$, of $\tr \proj \psi_S$.  One
interpretation of the proof of correctness of the product test is
therefore that, if the average entanglement of $\ket{\psi}$ across all
bipartite partitions of $[n]$ is low, as measured by the purity, then
$\ket{\psi}$ must in fact be close to a product state across all
subsystems.

In this appendix, we discuss a similar interpretation of our results in terms of
an average over product states, via the following proposition. 

\begin{prop}
\label{prop:prodoverlaps}
Given $\ket{\psi}\in (\bbC^d)^{\ot n}$,
\[ \Ptest(\proj{\psi}) = \left(\frac{d(d+1)}{2}\right)^n \bbE_{\ket{\phi_1},\dots,\ket{\phi_n}} \left[ |\braket{\psi}{\phi_1\dots\phi_n}|^4 \right] . \]
\end{prop}

\begin{proof}
Similarly to before, let the input to the product test be two copies $\psi_A$, $\psi_B$ of a state $\psi := \proj{\psi}$, and let $\mathcal{F}$ denote the swap operator that exchanges systems A and B. Then
\beas
\lefteqn{\bbE_{\ket{\phi_1},\dots,\ket{\phi_n}} \left[ |\braket{\psi}{\phi_1,\dots,\phi_n}|^4 \R]}\\
&=& \bbE_{\ket{\phi_1},\dots,\ket{\phi_n}} \left[ \tr (\psi_A \otimes \psi_B) ((\phi_1 \otimes \cdots \otimes \phi_n)_A \otimes (\phi_1 \otimes \cdots \otimes \phi_n)_B)\right] \\
&=& \tr (\psi_A \otimes \psi_B) \left(\bbE_{\ket{\phi}}\left[ \phi_A
    \otimes \phi_B \right]\right)^{\otimes n}\\ &=& \tr (\psi_A
\otimes \psi_B) \left(\frac{I+\mathcal{F}}{d(d+1)} \right)^{\otimes n}
= \left(\frac{2}{d(d+1)}\right)^n \Ptest(\proj{\psi}). 
\eeas
\end{proof}

We note that, in this interpretation, our main result is reminiscent of the so-called inverse theorem for the second Gowers uniformity norm \cite{gowers98,gowers01}, which we briefly outline. Let $f:\{0,1\}^n \rightarrow \bbR$ be some function such that $\frac{1}{2^n}\sum_x f(x)^2 = 1$, and let the $p$-norms of $f$ on the Fourier side be defined as
$\|\hat{f}\|_p = \left( \sum_{x \in \{0,1\}^n} \left| \frac{1}{2^n} \sum_{y \in \{0,1\}^n} (-1)^{x \cdot y} f(y) \right|^p \right)^{1/p}.$
Then it is straightforward to show that $\|\hat{f}\|_\infty^4 \le \|\hat{f}\|_4^4 \le \|\hat{f}\|_\infty^2$, where the quantity in the middle is known as the (fourth power of) the second Gowers uniformity norm of $f$. That is, $\|\hat{f}\|_\infty^2$ (representing the largest overlap of $f$ with a parity function) is well approximated by $\|\hat{f}\|_4^4$ (the {\em average} of the squared overlaps with parity functions). This simple approximation has proven useful in arithmetic combinatorics \cite{gowers98}.

Via the correspondence of Proposition \ref{prop:prodoverlaps}, Theorem
\ref{thm:prodtest} shows that a similar result holds if we replace the
cube $\{0,1\}^n$ with the space $(\bbC^d)^{\ot n}$: the largest
overlap with a product state can be well approximated by the average
squared overlap with product states. Note that if one attempts to use
the classical proof technique for the Gowers uniformity norm to prove
this result, one does not obtain Theorem \ref{thm:prodtest}, but a
considerably weaker result containing a term exponentially large in
$n$. Intuitively, this is because the set of overlaps with parity
functions for some function $f:\{0,1\}^n \rightarrow \bbR$ is
essentially arbitrary, whereas the set of overlaps of some state
$\ket{\psi}$ with product states is highly constrained. 

%-------------------------------------------------------------------------------

\section*{Acknowledgements}
We did most of this research while working at the University of Bristol.
AM was supported by the EC-FP6-STREP network QICS and an EPSRC
Postdoctoral Research Fellowship. AWH was supported by the EPSRC grant
``QIP-IRC'', by NSF
grants 0916400, 0829937, 0803478, DARPA QuEST contract FA9550-09-1-0044 and
the IARPA MUSIQC and QCS contracts.
 We would like to thank many people for inspiring discussions,
including Boaz Barak, Salman Beigi, Fernando Brand\~ao, Toby Cubitt, Sevag Gharibian, Leonid Gurvits, Julia Kempe, Hirotada Kobayashi,
Harumichi Nishimura, Thomas Vidick, Andreas Winter and Xiaodi Wu.
We would also like to thank the FOCS and JACM referees for their helpful comments.

%-------------------------------------------------------------------------------

\ifacm
\bibliographystyle{acmsmall}
\else
\bibliographystyle{abbrv}
\fi
\bibliography{thesis}

%-------------------------------------------------------------------------------

\end{document}